\documentclass[a4paper,11pt]{article}
\pdfoutput=1 % if your are submitting a pdflatex (i.e. if you have
\usepackage{jheppub} % for details on the use of the package, please
\usepackage[bottom]{footmisc}
\usepackage{amssymb}
\usepackage{amsmath}
\usepackage{amsthm}
\usepackage[usenames,dvipsnames]{xcolor}
\usepackage{epsfig}
\usepackage{dcolumn}
\usepackage{tikz}
\usetikzlibrary{shapes.geometric, arrows}
\usepackage{upgreek}
\usepackage{setspace}
\usepackage{enumitem}
\usepackage{array,multirow,bigdelim,arydshln}
\usepackage{appendix}
\usepackage{xparse}
\usepackage{stmaryrd}
\usepackage[T1]{fontenc} % if needed
\usepackage{mathtools}
\usepackage{physics} % Physiscs package
\usepackage{adjustbox}
\usepackage{multirow}
\usepackage{graphicx} % For images
\usepackage{float} % To manipulate figures
\graphicspath{{./images/}}
\usepackage[nottoc]{tocbibind}
\usepackage{hyperref}
\usepackage[utf8]{inputenc}
\hypersetup{
	colorlinks,
	urlcolor=Maroon,
	linkcolor=Maroon,
	citecolor=Maroon
	}

\NewDocumentCommand{\binomial}{omm}
 {%
  \genfrac(){0pt}{}{#2}{#3}%
  \IfValueT{#1}{_{\!#1}}%
 }
\NewDocumentCommand{\eulerian}{omm}
 {%
  \genfrac<>{0pt}{}{#2}{#3}%
  \IfValueT{#1}{_{\!#1}}%
 }

\def \s {\sigma}

\usepackage{latexsym}
\usepackage{tikz}

\theoremstyle{plain}

\newtheorem{thm}{Theorem}[section]

\newtheorem{lem}[thm]{Lemma}

\newtheorem{defn}[thm]{Definition}
\theoremstyle{definition}

\newcommand{\sfs}{\mathsf{s}}

\usepackage[percent]{overpic}
\usepackage{multirow} 
\usepackage{slashed}
\def\yz#1\yz {{\color{blue} [[YZ: #1]] }}
\def\bu#1\bu {{\color{red} [[BU: #1]] }}

\title{Planar Matrices and Arrays of Feynman Diagrams}

\author[a]{Freddy Cachazo,}\emailAdd{fcachazo@pitp.ca}
\author[a,b,c,d]{Alfredo Guevara,}\emailAdd{aguevaragonzalez@fas.harvard.edu}
\author[a,e]{Bruno Umbert,}\emailAdd{bgimenez@uwo.ca}
\author[f,a,g]{and Yong Zhang}\emailAdd{yzhang@perimeterinstitute.ca}

\affiliation[a]{Perimeter Institute for Theoretical Physics, Waterloo, ON N2L 2Y5, Canada}
\affiliation[b]{Department of Physics \& Astronomy, University of Waterloo, Waterloo, ON N2L 3G1, Canada}
\affiliation[c]{CECs Valdivia \& Departamento de F\'isica, Universidad de Concepci\'on, Casilla 160-C,\\ Concepci\'on, Chile}

\affiliation[d]{Center for the Fundamental Laws of Nature, Society of Fellows \& Black Hole Initiative, Harvard
University, Cambridge, MA 02138, USA}

\affiliation[e]{Department of Applied Mathematics, Western University, London, ON N6A 5B7, Canada}
\affiliation[f]{CAS Key Laboratory of Theoretical Physics, Institute of Theoretical Physics, Chinese Academy of Sciences, Beijing 100190, China}
\affiliation[g]{School of Physical Sciences, University of Chinese Academy of Sciences, No.19A Yuquan Road, Beijing 100049, China}

\abstract{Very recently planar collections of Feynman diagrams were proposed by Borges and one of the authors as the natural generalization of Feynman diagrams for the computation of $k=3$ biadjoint amplitudes. Planar collections are one-dimensional arrays of metric trees satisfying an induced planarity and compatibility condition. In this work we introduce planar matrices of Feynman diagrams as the objects that compute $k=4$ biadjoint amplitudes. These are symmetric matrices of metric trees satisfying compatibility conditions. We introduce two notions of combinatorial bootstrap techniques for finding collections from Feynman diagrams and matrices from collections. As applications of the first, we find all $693$, $13\,612$, and $346\,710$ collections for $(k,n)=(3,7), (3,8),$ and $(3,9)$ respectively. As applications of the second kind, we find all $90\, 608$  and $30\,659\,424$ planar matrices that compute $(k,n)=(4,8)$ and $(4,9)$ biadjoint amplitudes respectively. As an example of the evaluation of matrices of Feynman diagrams, we present the complete form of the $(4,8)$ and $(4,9)$ biadjoint amplitudes. We also start the study of higher dimensional arrays of Feynman diagrams, including the combinatorial version of the duality between $(k,n)$ and $(n-k,n)$ objects. 
}

\begin{document}

\maketitle

\addtocontents{toc}{\protect\setcounter{tocdepth}{1}}
\def \tr {\nonumber\\}
\def \la  {\langle}
\def \ra {\rangle}
\def\hset{\texttt{h}}
\def\gset{\texttt{g}}
\def\sset{\texttt{s}}
\def \be {\begin{equation}}
\def \ee {\end{equation}}
\def \ba {\begin{eqnarray}}
\def \ea {\end{eqnarray}}
\def \k {\kappa}
\def \h {\hbar}
\def \r {\rho}
\def \l {\lambda}
\def \be {\begin{equation}}
\def \en {\end{equation}}
\def \bes {\begin{eqnarray}}
\def \ens {\end{eqnarray}}
\def \red {\color{Maroon}}
\def \pt {{\rm PT}}
\def \s {\sigma} % Is there no problem with line 117?
\def \ls {{\rm LS}}
\def \ma {\Upsilon}
\def \s {\textsf{s}}
\def \t {\textsf{t}}
\def \R {\textsf{R}}
\def \W {\textsf{W}}
\def \U {\textsf{U}}
\def \e {\textsf{e}}

\numberwithin{equation}{section}

%\end{document}

\section{Introduction} \label{sec1}

Tree-level scattering amplitudes in a cubic scalar theory with flavor group $U(N)\times U(\tilde N)$ admit a Cachazo-He-Yuan (CHY) formulation based on an integration over the configuration space of $n$ points on $\mathbb{CP}^1$ and the scattering equations \cite{Fairlie:1972zz,Fairlie:2008dg,Cachazo:2013gna,Cachazo:2013hca,Dolan:2013isa}. In \cite{Cachazo:2019ngv}, Early, Mizera and two of the authors introduced generalizations to higher dimensional projective spaces $\mathbb{CP}^{k-1}$. These higher $k$ ``biadjoint amplitudes'' were also shown to have deep connections to tropical Grassmannians. This led to the proposal that Feynman diagrams could be identified with facets of the corresponding ${\rm Trop}\, G(k,n)$ \cite{Cachazo:2019ngv,speyer2004tropical,speyer2005tropical}.   Moreover, the generalized amplitudes and their properties, including generalized soft/hard theorems, and the corresponding scattering equations, including characterizations of solutions, have been further studied in \cite{Cachazo:2019apa,Sepulveda:2019vrz,Cachazo:2019ble,GarciaSepulveda:2019jxn,Abhishek:2020xfy,Early:2022mdn,Cachazo:2021wsz}.  

Motivated by the connection between ${\rm Trop}\, G(2,n)$ with metric trees, which can be identified as Feynman diagrams, and ${\rm Trop}\, G(3,n)$ with metric arrangements of trees \cite{herrmann2009draw}, Borges and one of the authors introduced a generalization to $k=3$ called planar collections of Feynman diagrams as the objects that compute $k=3$ biadjoint amplitudes \cite{Borges:2019csl}. 

The computation of a $k=3$ biadjoint amplitude is completely analogous to that of the standard $k=2$ amplitude but defined as a sum over planar collections of Feynman diagrams 
\be\label{qori}
m_n^{(3)}(\alpha,\beta) = \sum_{{\cal C}\in \Omega(\alpha) \cap \Omega(\beta)} {\cal R}({\cal C}),
\ee
with $\Omega(\alpha)$ the set of all collections of Feynman diagrams which are planar with respect to the $\alpha$-ordering \cite{Borges:2019csl}. More explicitly, the $i^{\rm th}$ tree in a collection is a tree with $n-1$ leaves $\{1,2,\ldots ,n\}\setminus i$ which is planar with respect to the ordering induced by deleting $i$ from $\alpha$. This is why the collection is called {\it  planar} and not the individual trees.

The value ${\cal R}({\cal C})$ of a planar collection ${\cal C}$ is obtained from the following function 
\be\label{fft}
{\cal F}({\cal C}) = \sum_{i,j,k}\pi_{ijk}\, \sfs_{ijk}
\ee
defined in terms of the metrics of the trees in the collection $d^{(i)}_{jk}$ which satisfy a compatibility condition $d^{(i)}_{jk}=d^{(j)}_{ik}=d^{(k)}_{ij}$, thus defining a completely symmetric rank three tensor $\pi_{ijk}:=d^{(i)}_{jk}$ \cite{herrmann2009draw}. Here $\sfs_{ijk}$ is the $k=3$ generalization of Mandelstam invariants, defined as completely symmetric rank-three tensors satisfying \cite{Cachazo:2019ngv}
\be\label{1333}
\sfs_{iij}=0, \quad \sum_{j,k=1}^n \sfs_{ijk}=0, \quad \forall i\in \{1,2,\ldots ,n\}.
\ee
The explicit value is then computed as
\be\label{intew}
 {\cal R}({\cal C})=\int_{\Delta}d^{2(n-4)}f_I \, {\rm exp}\, {\cal F}({\cal C})\,,
\ee
where the domain $\Delta$ is defined by the requirement that all internal lengths of all Feynman diagrams in the collection be positive \cite{Borges:2019csl}.

In this work we continue the study of planar collections of Feynman diagrams by exploiting an algorithm proposed in \cite{Borges:2019csl} for determining all collections for $k=3$ and $n$ points by a ``combinatorial bootstrap'' starting from $k=2$ and $n$-point planar Feynman diagrams. We review in detail the algorithm in section \ref{sec2} and use it to construct all $693$, $13\,612$, and $346\,710$ collections for $(k,n)=(3,7)$, $(3,8)$ and $(3,9)$ respectively. The $693$ collections for $(k,n)=(3,7)$ were already obtained in \cite{Borges:2019csl} by imposing a planarity condition on the metric tree arrangements presented by Herrmann, Jensen, Joswig, and Sturmfels  in their study of the tropical Grassmannian ${\rm Trop}\, G(3,7)$ in \cite{herrmann2009draw}. Also, there are deep connections between positive tropical Grassmannians and cluster algebras as explained by Speyer and Williams in \cite{SpeyerW} and explored by Drummond,  Foster,  Gürdogan, and Kalousios in  \cite{Drummond:2019qjk}. In the latter work  it was found that ${\rm Trop}^+\, G(3,8)$ can be described in terms of $25\,080$ clusters. Here we show that our $13\,612$ planar collections for $(3,8)$ encode exactly the same information as their $25\,080$ clusters. The cluster algebra analysis of ${\rm Trop}^+\, G(3,9)$ has not appeared in the literature but it should be possible to obtain them from our $346\,710$ collections.

We also start the exploration of the next layer of generalizations of Feynman diagrams in section \ref{sec3} and propose that $k=4$ biadjoint amplitudes are computed using planar matrices of Feynman diagrams. In a nutshell, an $n$-point planar matrix of Feynman diagrams ${\cal M}$ is an $n\times n$ matrix with Feynman diagrams as entries. The ${\cal M}_{ij}$ entry is a Feynman diagram with $n-2$ leaves $\{1,2,\ldots ,n\}\setminus \{i,j\}$. Each tree has a metric defined by the minimum distance between leaves, $d^{(ij)}_{kl}$. Here we use superscripts to denote the entry in the matrix of trees and subscripts for the two leaves whose distance is given. Planar matrices of Feynman diagrams must satisfy a compatibility condition on the metrics
\be\label{comp4}
d^{(ij)}_{kl} = d^{(ik)}_{jl}=d^{(il)}_{kj} =d^{(kl)}_{ij} = d_{ik}^{(jl)}=d_{il}^{(kj)}.
\ee
This means that the collection of all metrics defines a completely symmetric rank four tensor $\pi_{ijkl}:= d^{(ij)}_{kl}$.

Using this we generalize the prescription for computing the value $R(T)$ and ${\cal R}({\cal C})$ of $k=2$ and $k=3$ ``diagrams'' to ${\cal R}({\cal M})$ for $k=4$ and therefore their contribution to generalized $k=4$ amplitudes.

Moreover, we find that a second class of combinatorial bootstrap approach can be efficiently used to simplify the search for matrices of diagrams that satisfy the compatibility conditions \eqref{comp4}. The idea is that any column of a planar matrix of Feynman diagrams must also be a planar collection of Feynman diagrams but with one less particle. In the first of our two main examples, any matrix for $(k,n)=(4,8)$ must have columns taken from the set of $693$ $(k,n)=(3,7)$ planar collections. Using that the matrix must be symmetric, one can easily find $91\, 496$ matrices of trees satisfying this purely combinatorial condition. Therefore the set of all valid planar matrices for $(k,n)=(4,8)$ must be contained in the set of those $91\, 496$ matrices. Surprisingly, we find that only $888$ such matrices do not admit a generic metric satisfying \eqref{comp4}. This means that there are exactly $90\, 608$ planar matrices of Feynman diagrams for $(k,n)=(4,8)$. We also find efficient ways of computing their contribution to $m_8^{(4)}(\mathbb{I},\mathbb{I})$. 

As the second main example of the technique, we use the $(3,8)$ planar collections to construct candidate matrices in $(4,9)$. We find $33\,182\,763$ such symmetric objects. Computing their metrics we find that $2\, 523\, 339$ of them are degenerate and therefore the total number of planar matrices of Feynman diagrams for $(4,9)$ is $30\, 659 \, 424$. We present all results, including the amplitudes, in ancillary files
% \footnote{The ancillary files are \textsc{Mathematica} notebooks while the data together with the notebooks are contained in text files available at \url{https://www.dropbox.com/sh/w5i3vhig1qm1r0f/AAAuF-vtRxUCFRj5BTAqiF9wa?dl=0}.} 
and explain the results in section \ref{sec4}.

In section \ref{sec5} we explain how to use efficient techniques for evaluating the contribution of a given planar array of Feynman diagrams to an amplitude by showing that the integration over the space of metrics is equivalent to the triangulations of certain polytopes and then show how softwares such as PolyMake can be used to carry out the computations.

After identifying collections with $(3,n)$ amplitudes and matrices with $(4,n)$, it is natural to introduce planar $(k-2)$-dimensional arrays of Feynman diagrams as the objects relevant for the computation of $(k,n)$ biadjoint amplitudes. In section \ref{sec6} we discuss these objects and explain the combinatorial version of the duality connecting $(k,n)$ and $(n-k,n)$ biajoint amplitudes at the level of the arrays.

This paper is organized as follows. We explain two combinatorial bootstrap techniques in sections \ref{sec2} and \ref{sec3} respectively, with data gathered and explained in section \ref{sec4}. In section \ref{sec5}, we show how to evaluate the planar arrays of Feynman diagrams as partial amplitudes efficiently. Their duality is discussed in section \ref{sec6}.
We end in section \ref{sec7} with discussions and future directions. 
Further details that
complement the main text can be found in the appendices. Most data is presented in ancillary files.

\section{Planar Collections of Feynman Diagrams} \label{sec2}

In this section we give a short review of the definition and properties of planar collections of Feynman diagrams \cite{Borges:2019csl}. Emphasis is placed on a technique for constructing $n$-particle planar collections starting from special ones obtained by ``pruning'' $n$-point planar Feynman diagrams and then applying a ``mutation'' process. Here we borrow the terminology {\it mutation} from the cluster algebra literature \cite{ClusterA,ClusterB,ClusterC}. The reason for this becomes clear below.  

This pruning-mutating technique is the first combinatorial bootstrap approach we use in this work. The second kind is introduced in section \ref{sec3} as a way of constructing planar matrices of Feynman diagrams from planar collections.

\begin{figure}[!htb]
\centering

\includegraphics[width=150mm]{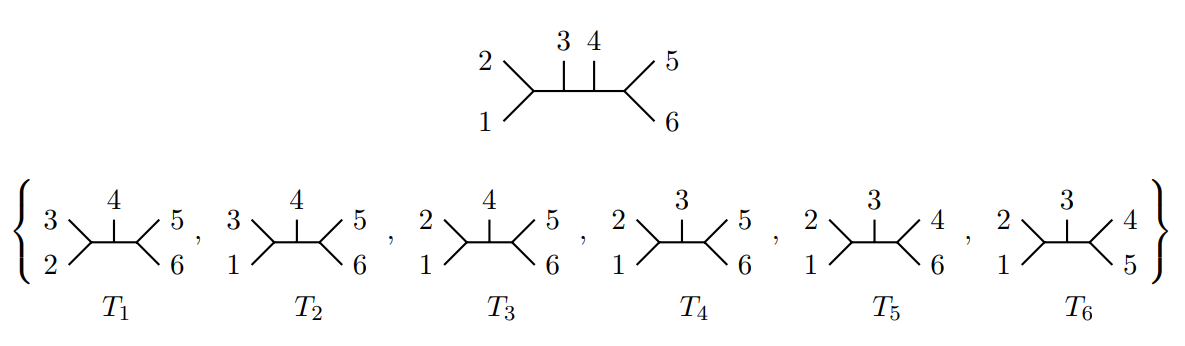}
\caption{An example for an initial planar collection obtained by pruning a 6-point Feynman diagram. Above is the 6-point Feynman diagram to be pruned. Below is the planar collection of 5-point Feynman diagrams obtained by pruning the leaves $1,2,\cdots, 6$ of the above Feynman diagram respectively.}
\label{fig:universe2}
\end{figure}

Without loss of generality, from now on we only consider the canonical ordering $\mathbb{I}:=(1,2,\ldots ,n)$ and every time an object is said to be planar, it means with respect to $\mathbb{I}$.

Recall that for $k=2$ the objects of interest are $n$-particle planar Feynman diagrams in a $\phi^3$ scalar theory. There are exactly $C_{n-2}$ such diagrams\footnote{$C_m$ is the $m^{\rm th}$ Catalan number.}. When Feynman diagrams are thought of as metric trees, a length is associated to each edge and if any of the $n-3$ internal lengths becomes zero we say that the tree degenerates. Here is where the power of restricting to planar objects comes into play, once a given planar tree degenerates, there is exactly one more planar tree that shares the same degeneration. These two planar Feynman diagrams only differ by a single pole and we say that they are related by a mutation. 

Starting from any planar Feynman diagram, one can get all other planar Feynman diagrams by repeating mutations. If no new Feynman diagrams are generated, we are sure we have obtained all of the Feynman diagrams of certain ordering.  

Here the terminology {\it mutation} precisely coincides with the one used in cluster algebras since planar Feynman diagrams are known to be in bijection with clusters of an $A$-type cluster algebra and mutations connect clusters in exactly the same way as degenerations connect planar metric trees. This precise connection between objects connected via degenerations and cluster mutations does not hold for higher $k$ and therefore we hope the abuse of terminology will not cause confusion \cite{Borges:2019csl}.

For the computation of $k=3$ biadjoint amplitudes, planar $n$-point Feynman diagrams are replaced by  planar collections of $(n-1)$-point Feynman diagrams. Each collection is made out of $n$ Feynman diagrams with the $i^{\rm th}$ tree defined on the set $\{1,2,\cdots,n\}\setminus i$ and planar the respect to the ordering $(1,2,\cdots, i-1,i+1,\cdots, n)$. Each tree has its own metric defined as the matrix of minimal lengths from one leaf to another. The metric for the $i^{\rm th}$-tree is denoted as $d^{(i)}_{jk}$ with $j,k \in \{1,2,\cdots,n\}\setminus i$. Moreover, the metrics have to satisfy a compatibility condition $d^{(j)}_{kl}=d^{(k)}_{jl}=d^{(l)}_{jk}$.

A necessary condition for two planar collections of Feynman diagrams to be related is that their individual elements, i.e. the $(n-1)$-point Feynman diagrams, are either related by a mutation or are the same. Of course, in order to prove that the collections are actually related it is necessary to study the space of metrics and show that the two share a common degeneration.    

The key idea is that we can get all planar collections of Feynman diagrams by repeated mutations, starting at any single collection. What is more, we can tell whether we have obtained all of the collections when there are no new collections produced by mutations \footnote{Here we assume that the set of all of planar collections is connected. We have checked this to be the case up to $n=9$.}.   

A more efficient variant of the mutation procedure described above is obtained by introducing multiple initial collections. In fact there is a canonical set of planar collections which are easily obtained from $n$-point planar Feynman diagrams. 

Let us define the {\it initial planar collections} as those obtained via the following procedure. Consider any $n$-point planar Feynman diagram $T$ and denote the tree obtained by pruning (or removing) the $i^{\rm th}$ leaf by $T_i$. Then the set $\{T_1,T_2,\ldots ,T_n\}$ is a planar collection of Feynman diagrams. 

Let us illustrate this with a simple example seen in figure \ref{fig:universe2}.

%\iffalse
%\begin{figure}[!htb]
%\centering
%\begin{overpic}[scale=0.2]{incol.jpg}
%\put(-4,10){$\Bigg\{$}
%\put(102,10){$\Bigg\}$}
%\put(8,-1){$T_1$}
%\put(28,-1){$T_2$}
%\put(48,-1){$T_3$}
%\put(68,-1){$T_4$}
%\put(88,-1){$T_5$}
%\end{overpic}
%\caption{An example for a 5-pt initial planar collection. Above is a 5-pt Feynman diagram. Below is a collection of 4-pt Feynman diagrams obtained by pruning the leaves $1,2,\cdots, 5$ of the above Feynman diagram respectively.}
%\label{fig:universe2}
%\end{figure}
%\fi   

\begin{table}[!htb]
\centering
\begin{tabular}{c|c|c|c|c|c|c|c}
 \hline
 $(k,n)$&${\rm Number~ of~ \atop collections}$& \multicolumn{5}{c|}{${\rm Numbers~ of~  collections~ for~ each ~kind}$ } &${\rm Number~ of~ \atop layers}$   \\ 
 \hline
\multirow{2}{*}{(3,5)} & \multirow{2}{*}{5}  & 2-mut. &\multicolumn{4}{c|}{ } & \multirow{2}{*}{0}    \\
 &  & 5 &\multicolumn{4}{c|}{ } & \\
  \hline
  \multirow{2}{*}{(3,6)} & \multirow{2}{*}{48}  & 4-mut. & 6-mut. &\multicolumn{3}{c|}{ } & \multirow{2}{*}{3}    \\
 &  & 46 & 2&\multicolumn{3}{c|}{ } & \\
  \hline
  \multirow{2}{*}{(3,7)} & \multirow{2}{*}{693}  & 6-mut. & 7-mut. & 8-mut. &\multicolumn{2}{c|}{ } & \multirow{2}{*}{4}    \\
 &  &  595&  28 & 70 &\multicolumn{2}{c|}{ } & \\
  \hline
  
    \multirow{2}{*}{(3,8)} & \multirow{2}{*}{ 13 612}  & 8-mut. & 9-mut. & 10-mut.  & 11-mut. & 12-mut. & \multirow{2}{*}{8}   \\
 &  &  9 672&  1 488 & 2 280 & 96 &76 \\
  \hline
  
  \multirow{4}{*}{(3,9)} & \multirow{4}{*}{346 710}  & 10-mut.  & 11-mut.  & 12-mut.  & 13-mut. & 14-mut. & \multirow{4}{*}{11}    \\
 &  &  186 147&  61 398& 78 402 & 12 300& 7 668 \\
\cline{3-7}
   &  &  15-mut.  & 16-mut.  & 17-mut. &\multicolumn{2}{c|}{ } &  \\
   &  & 522 & 270 & 3 &\multicolumn{2}{c|}{ } &  \\
  \hline
\end{tabular}
\caption{\label{3ntable} Summary of results for planar collections of Feynman diagrams for $k=3$ and up to $n=9$. The second column gives the total numbers of planar collections. The third column provides the numbers of collections for each kind, classified by the number of mutations. The fourth column indicates how many layers of mutations are necessary to find the complete set of collections starting with the $C_{n-2}$ initial collections.}
\end{table}

Using all such $C_{n-2}$ collections as starting points one can then apply mutations to each and start filling out the space of planar collections in $n$-points. When the method is applied to $(k,n)=(3,5)$ we obtain all planar collections without the need of any mutations since every single planar collection in this case is dual to a $(2,5)$ Feynman diagram. Next, we apply the technique to reproduce the known results for $(3,6)$ starting from the $C_4=14$ initial collections. 

We find that after only three layers of mutations we get all planar collections. Repeating the procedure for $(3,7)$ we find all $693$ planar collections stating from the initial $C_5=42$ collections after four layers of mutations. 

Our first new results in this work are the computation of all $13\, 612$ planar collections in $(3,8)$ and all $346\, 710$ in $(3,9)$. Details on the results and the ancillary files where the collections are presented are provided in section \ref{sec4}.

All results are summarized in table \ref{3ntable}. We classify the planar collections according to their numbers of mutations and count the numbers of collections for each kind as well. The precise definition of metrics and degenerations of planar collections of Feynman diagrams was given in \cite{Borges:2019csl}.

\section{Planar Matrices of Feynman Diagrams} \label{sec3}

In the previous section, we introduced an efficient algorithm for finding all planar collections of Feynman diagrams based on a pruning-mutation procedure. Such collections compute $k=3$ biadjoint amplitudes. 

The next natural question is what replaces planar collections for $k=4$ biadjoint amplitudes. Inspired by the way a single planar Feynman diagram defines a collection by pruning one leaf at a time, we start with a matrix of Feynman diagrams where the $i,j$ element is obtained by pruning the $i^{\rm th}$ and $j^{\rm th}$ leaves of an n-point planar Feynman diagrams as the relevant objects for $k=4$, 
\be\label{matrix}
{\cal M}=
\begin{bmatrix} 
    \emptyset & T^{(1,2)} & \dots & T^{(1,n-1)} & T^{(1,n)} \\
    T^{(2,1)} & \emptyset & \dots & T^{(2,n-1)} & T^{(2,n)} \\
    \vdots & \vdots & \ddots & \vdots & \vdots\\
   T^{(n-1,1)} & T^{(n-1,2)} & \dots & \emptyset & T^{(n-1,n)}\\
   T^{(n,1)} & T^{(n,2)} & \dots & T^{(n,n-1)} & \emptyset 
    \end{bmatrix}\,,
\ee
as first proposed in \cite{Borges:2019csl}.  We denote the Feynman diagram in the $i^{\rm th}$ row and $j^{\rm th}$ column, where labels $i$ and $j$ are absent, by $T^{(i,j)}$.  

We add a metric to every Feynman diagram $T^{(i,j)}$ in the matrix, and  denote the lengths of internal and external edges as $f_{I}^{(ij)}$ and $e_{m}^{(ij)}$ respectively.  Correspondingly,  we can use $d_{kl}^{(ij)}$ to denote the minimal  distance between two leaves $k$ and $l$. Up to this point, the edge lengths and hence distances $d^{(ij)}_{kl}$ of different Feynman diagrams in the matrix have no relations. We can relate them by imposing compatibility conditions analogous to those for collections of Feynman diagrams. This leads to the following definition.

\begin{defn}
A {\it planar matrix of Feynman diagrams} is an $n\times n$ matrix ${\cal M}$ with component ${\cal M}_{ij}$ given by a metric tree with leaves $\{1,2,\ldots ,n\}\setminus \{i,j\}$ and planar with respect to the ordering $(1,2,\cdots, \slashed{i},\cdots, \slashed{j},\cdots, n)$ satisfying the following conditions
\begin{itemize}
    \item  Diagonal entries are the empty tree ${\cal M}_{ii}=\emptyset$.
    \item  Compatibility \eqref{comp4} 
$$d^{(ij)}_{kl} = d^{(ik)}_{jl}=d^{(il)}_{kj} =d^{(kl)}_{ij} = d_{ik}^{(jl)}=d_{il}^{(kj)}.$$
\end{itemize}

\end{defn}

Note that the compatibility condition has several important consequences. The first is that since a given metric is symmetric in their labels, i.e. $d_{kl}^{(ij)}=d_{lk}^{(ij)}$ which is obvious from its definition as the minimum distance from $k$ to $l$, one finds that the matrix ${\cal M}$ must be symmetric as stated in the following lemma. 

\begin{lem}
\label{symmetric}
Planar matrices of Feynman diagrams are symmetric.
\end{lem}

\begin{proof}
The symmetry of the matrix follows from realizing that the compatibility condition requires that $d_{kl}^{(ij)}=d_{ij}^{(kl)}$ and therefore the symmetry of the metric on the lhs in the leave labels $k$ and $l$ implies that of the rhs is symmetric in the matrix labels $k$ and $l$. In order to complete the proof, it is enough to note that a binary metric tree is uniquely determined by its metric as we show in appendix \ref{apA}. 
\end{proof}

Planar collections of Feynman diagrams have $(n-4)n$ internal edges; $n-4$ for each of the $n$ trees in the collection. However, only $2(n-4)$ are independent once the compatibility condition is imposed on the metrics as reviewed in \cite{Borges:2019csl}. In the case of planar matrices of Feynman diagrams there are ${n\choose2} (n-4)$ internal lengths $f_{I}^{(ij)}$ with $1\leq i<j\leq n, 1\leq I\leq n-5$ while the compatibility conditions \eqref{comp4} reduce the number down to $3(n-5)$ independent ones. This means that a planar matrix has at least $3(n-5)$ possible degenerations. The precise number depends on the structure of the trees in the matrix. 

In analogy with planar collections, we say that two planar matrices are related via a mutation if they share a co-dimension one degeneration.

Recall that an initial planar collection is obtained by pruning a leaf of the same $n$-point planar Feynman diagram to produce $n$ different $(n-1)$-point trees.  We can also get an {\it initial planar matrix} by pruning two different leaves at a time from the same $n$-point planar Feynman diagram. See figure \ref{universe3} for an example. 
\begin{figure}[!htb]
\centering
\includegraphics[width=150mm]{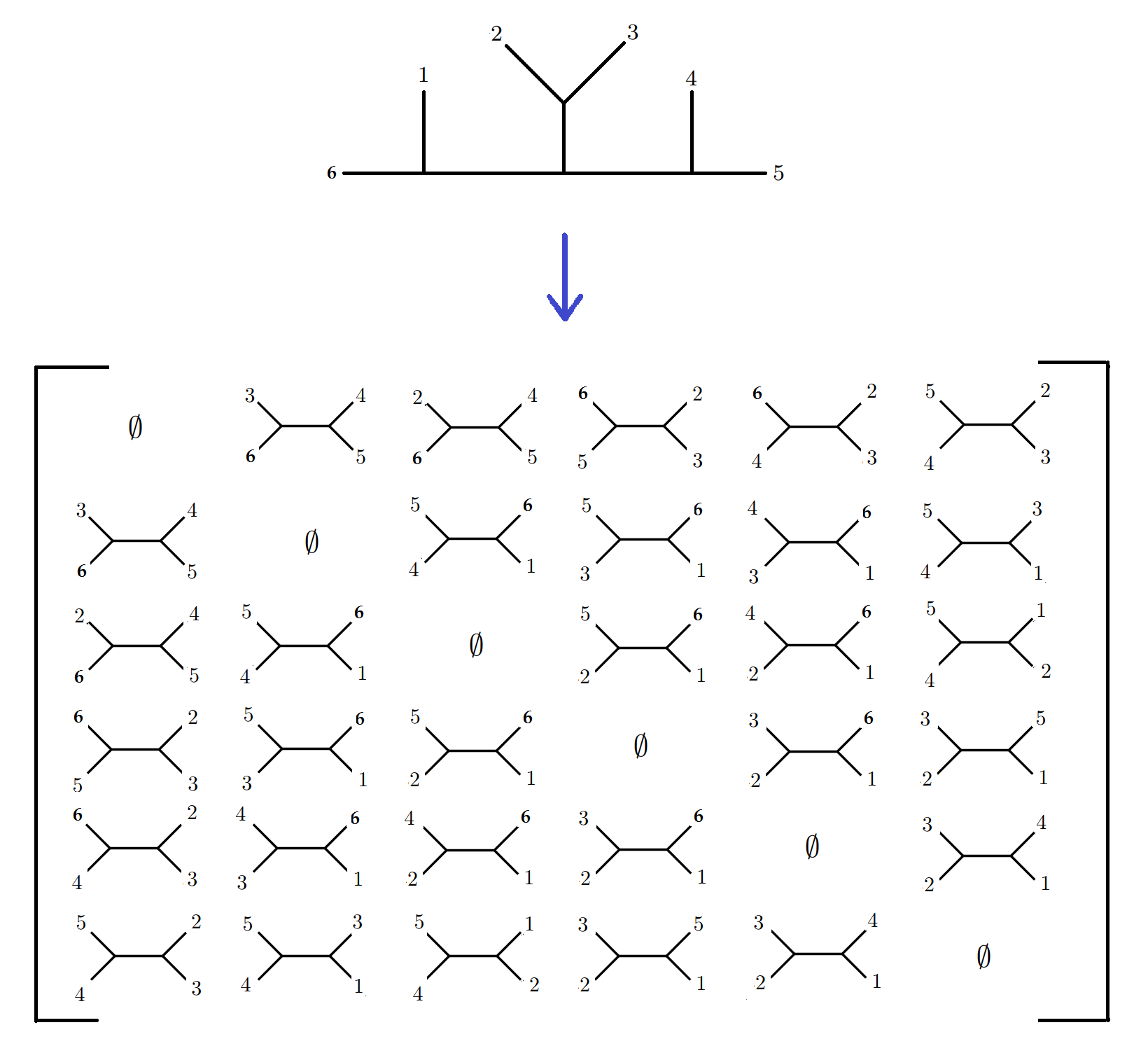}
\caption{An example for a 6-point initial planar matrix. Above we show a 6-point Feynman diagram. Below there is a symmetric matrix of 4-point Feynman diagrams obtained by pruning two leaves from the set $1,2,\cdots, 6$ at a time of the above Feynman diagram. The Feynman diagram from the $i^{\rm th}$ column and  $j^{\rm th}$ row has the $i^{\rm th}$ and $j^{\rm th}$ leaves pruned. \label{universe3}}
\end{figure}

%\begin{figure}[!htb]
%\centering
%\begin{overpic}[scale=0.15]{inma.jpg}
%\end{overpic}
%\caption{\label{fig:universe3} An example for a 6-point initial planar matrix. Above is a 6-point Feynman diagram. Below is a symmetric matrix of 4-point Feynman diagrams obtained by pruning two leaves $1,2,\cdots, 6$ at a time of the above Feynman diagram. The Feynman diagram from the $i^{\rm th}$ column and  $j^{\rm th}$ row has the $i^{\rm th}$ and $j^{\rm th}$ leaves pruned. }
%\end{figure}
Using all such $C_{n-2}$ matrices as starting points one can then apply mutations to each and start filling out the space of planar matrices in $n$-points. 

The contribution to the amplitudes of every planar matrix can be calculated individually.   
Consider the function of a planar matrix of Feynman diagrams ${\cal M}$,
\be
{\cal F}({\cal M}):= \sum_{1\leq i,j,k,l\leq n} \pi_{ijkl}\, \sfs_{ijkl},
\ee
with $\pi_{ijkl}:=d^{(ij)}_{kl}$. Here $\sfs_{ijkl}$ are the generalized symmetric Mandelstam invariants introduced in \cite{Cachazo:2019ngv}. These satisfy the conditions
\be\label{condS}
 \sfs_{iijk}=0, \qquad  \sum_{j,k,l=1}^n \sfs_{ijkl} =0 \quad  \forall i.
\ee
At this point it is not obvious but these conditions make it possible to write ${\cal F}({\cal M})$ in a form free of any length of leaves $e_{m}^{(ij)}$. In section \ref{sec6} we explain this phenomenon in more generality for any value of $k$.

An integral of ${\cal F}({\cal M})$ over independent internal lengths  $\{f_{1},f_{2},\cdots,f_{3(n-5)}\}$  gives the contribution to $k=4$ biadjoint amplitudes
\be 
 {\cal R}({\cal M})=\label{intFTM}
 \int_{\Delta}d^{3(n-5)}f_I \, {\rm exp}\, {\cal F}({\cal M})\,,
\ee
where the domain $\Delta$ is defined by the condition that all ${n\choose2} (n-4)$ internal lengths are positive and not only the $3(n-5)$ independent ones. For future use we comment that it is possible to consider \eqref{intFTM} also for degenerate matrices and in such cases it integrates to zero as its domain is a set of measure zero.  

Another important observation is that, in the $j$-th column or row of a planar matrix, all Feynman diagrams are free of particle $j$ and
the compatibility condition \eqref{comp4} requires
\be\label{comp444}
d^{(ij)}_{kl} = d^{(kj)}_{li}=d^{(lj)}_{ik},
\ee
for every three different particles $i,k,l$ of the remaining $n-1$ particles.
This means the $j$-th column or row is nothing but a planar collection of Feynman diagrams. Each column of a planar matrix is therefore made out of planar collections of $(3,n-1)$. Besides, once several columns have been fixed, the remaining columns have much less choices because of the symmetry requirement of  the matrix. This simple but powerful observation leads to the second kind of combinatorial bootstrap, which we describe next.

\subsection{Second Combinatorial Bootstrap} \label{sec31}

Suppose we have obtained all of the $N$ planar collections for the ordering $(1,2,\cdots,n-1)$. Let us denote the set of all such collections as $E_{3,n-1}= \{{\cal C}_1,{\cal C}_2,\cdots, {\cal C}_N \}$. The last column $\{ T^{(1,n)},T^{(2,n)},\cdots, T^{(n-1,n)}\}$ (here we have omitted the trivial empty tree $\emptyset$) of any planar matrix ${\cal M}$, where by definition particles $1,2,\cdots,n-1$ are deleted respectively in addition to the common missing particle $n$, must be an element of $E_{3,n-1}$.  

Now we consider a cyclic permutation with respect to the order $(1,2,\cdots,n-1,n)$ of particle labels of the set $E_{3,n-1}$,   
\be\label{cycper}
E_{3,n-1}^{(a)}=\{{\cal C}_1^{(a)},{\cal C}_2^{(a)},\cdots, {\cal C}_N ^{(a)} \}:= E_{3,n-1}{\big|}_{i \to i+a}\,.
\ee 
Clearly, particle labels are to be understood modulo $n$. One can see that $E_{3,n-1}^{(a)}$ is the set of all planar collections for the ordering $(1,2,\cdots, a-1, a+1, \cdots  ,n)$ with particle $a$ absent. By definition, we have $E_{3,n-1}^{(n)}\equiv E_{3,n-1}^{(0)}\equiv E_{3,n-1}$. The $a$-th column  $\{ T^{(1,a)},T^{(2,a)},\cdots, T^{(a-1,a)},T^{(a+1,a)},\cdots, T^{(n-1,a)}\}$ (here we have once again omitted the trivial tree $\emptyset$) of a planar matrix ${\cal M}$ must belong to the set $E_{3,n-1}^{(a)}$. Thus any planar matrix of Feynman diagrams must take the form
\be\label{formm}
{\cal M}= [ {\cal C}_{i_1}^{(1)},  {\cal C}_{i_2}^{(2)}, \cdots, {\cal C}_{i_n}^{(n)} ] \,, \quad \text{with}~ 1\leq i_1,\cdots,i_n \leq N\,. 
\ee 
Naively, we have $N$ choices for each column and hence $N^n$ candidate planar matrices. In principle, one could take this set of $N^n$ matrices and impose the compatibility condition on the metrics thus reducing the set to that of all planar matrices of Feynman diagrams. However, this procedure is impractical already for $n=7$ where $N=693$. 

Luckily, according to the Lemma \ref{symmetric}, the symmetry requirement of a planar matrix reduces this number dramatically. It is much more efficient to find possible planar matrices from all of the symmetric matrices of the form \eqref{formm}.  

Using this method we have obtained all planar matrices up to $n=9$. 
Table \ref{k4matrices} is a summary of our results.

\begin{table}[!htb]  
  \centering
  \begin{tabular}{ |p{3.5cm}|p{2.0cm}|p{2.0cm}|p{2.0cm}|p{2.0cm}|}
    \cline{2-5} 
    \multicolumn{1}{c|}{} &\hspace{5mm}\textbf{(4,\,6)} &\hspace{5mm}\textbf{(4,\,7)}
    &\hspace{5mm}\textbf{(4,\,8)}&\hspace{5mm}\textbf{(4,\,9)}\\
 \hline
 Planar Matrices   & \hspace{7mm}14 & \hspace{7mm}693 & \hspace{5mm}90\,608 & \hspace{2mm}30\,659\,424\\
  \hline
Degenerate Matrices & \hspace{8mm}0 & \hspace{8mm}0 & \hspace{7mm}888 & \hspace{3mm}2\,523\,339\\
  \hline
  \end{tabular}
  \caption{Number of planar matrices of Feynman diagrams and number of degenerate matrices for different values of $n$.}
  \label{k4matrices}
  \end{table}
  
More explicitly, when this method is applied to $(k,n)=(4,6)$, we obtain exactly all $14$ planar matrices, which are dual to the $C_4=14$ planar Feynman diagrams of $(2,6)$.  As there are $693$ planar collections in $(3,7)$, the duality between $(3,7)$ and $(4,7)$ implies that there should be $693$ planar matrices in (4,7) as well. In fact, our combinatorial bootstrap procedure results in exactly that number! Moreover, in section \ref{sec6} we explain how the $693$ planar matrices of Feynman diagrams map one to one onto the $693$ planar collections via the duality. 

Our second set of new results corresponds to the more interesting cases of $(4,8)$ and $(4,9)$, where our procedure leads to $91\, 416$ and $33\, 182\, 763$ symmetric matrices respectively. 

Having the set of all possible candidate matrices, we further determined that $90\, 608$ and $30\, 659\, 424$ of them respectively satisfy the compatibility conditions \eqref{comp4} while not becoming degenerate and thus get these numbers of planar matrices of Feynman diagrams. We see that in both cases, the combinatorial bootstrap came very close to the correct answer. We comment that the extra $888$ and 
$2\, 523\, 339$ ``offending'' symmetric matrices are actually degenerate planar matrices. This means that if we were to use all matrices obtained from the bootstrap in the formula for the amplitude we would still get the correct answer since the extra matrices integrate to zero under the formula \eqref{intFTM}. So we can just use all of the symmetric matrices to calculate the biadjoint amplitudes for $k=4$ as well.  

Below we show two explicit examples for $n=6,7$ in order to illustrate the procedure. These examples show why this is an efficient technique for getting planar matrices from collections of $(3,n-1)$. Details on the results for $(4,8)$ and $(4,9)$ and the ancillary files where the collections are presented are provided in section \ref{sec4}.

\subsection{A Simple Example: From $(3,5)$ to $(4,6)$} \label{sec32}

Now we proceed to show an explicit example of how to obtain planar matrices of Feynman diagrams for $(4,6)$. 
In this example, given the duality $(4,6)\sim (2,6)$, we could obtain the planar matrices by picking $n=6$ Feynman diagrams in $k=2$ and remove two leaves in a systematic way as shown in figure \ref{universe3}. Here, however, we introduce an algorithm to get the matrices using a second bootstrap approach constrained by the consistency conditions explained above, thus obtaining planar matrices from planar collections of Feynman diagrams of $(3,5)$. 

This algorithm works for general $n$, i.e. it obtains planar matrices of $(4,n)$ from planar collections of $(3,n-1)$, and is going to be particularly useful for larger $n$, where the number of matrices is considerably large. 

Before going through the algorithm, let's review the planar collections of $(3,5)$. There are 5 planar collections  of Feynman diagrams $E_{3,5}=\{{\cal C}_1,{\cal C}_2,{\cal C}_3,{\cal C}_4,{\cal C}_5\}$ \cite{Borges:2019csl}. These can come from the caterpillar tree in $n=5$ and its 4 cyclic permutations, shown in figure \ref{35col}.  

\begin{figure}[!htb]
\centering
\includegraphics[width=140mm]{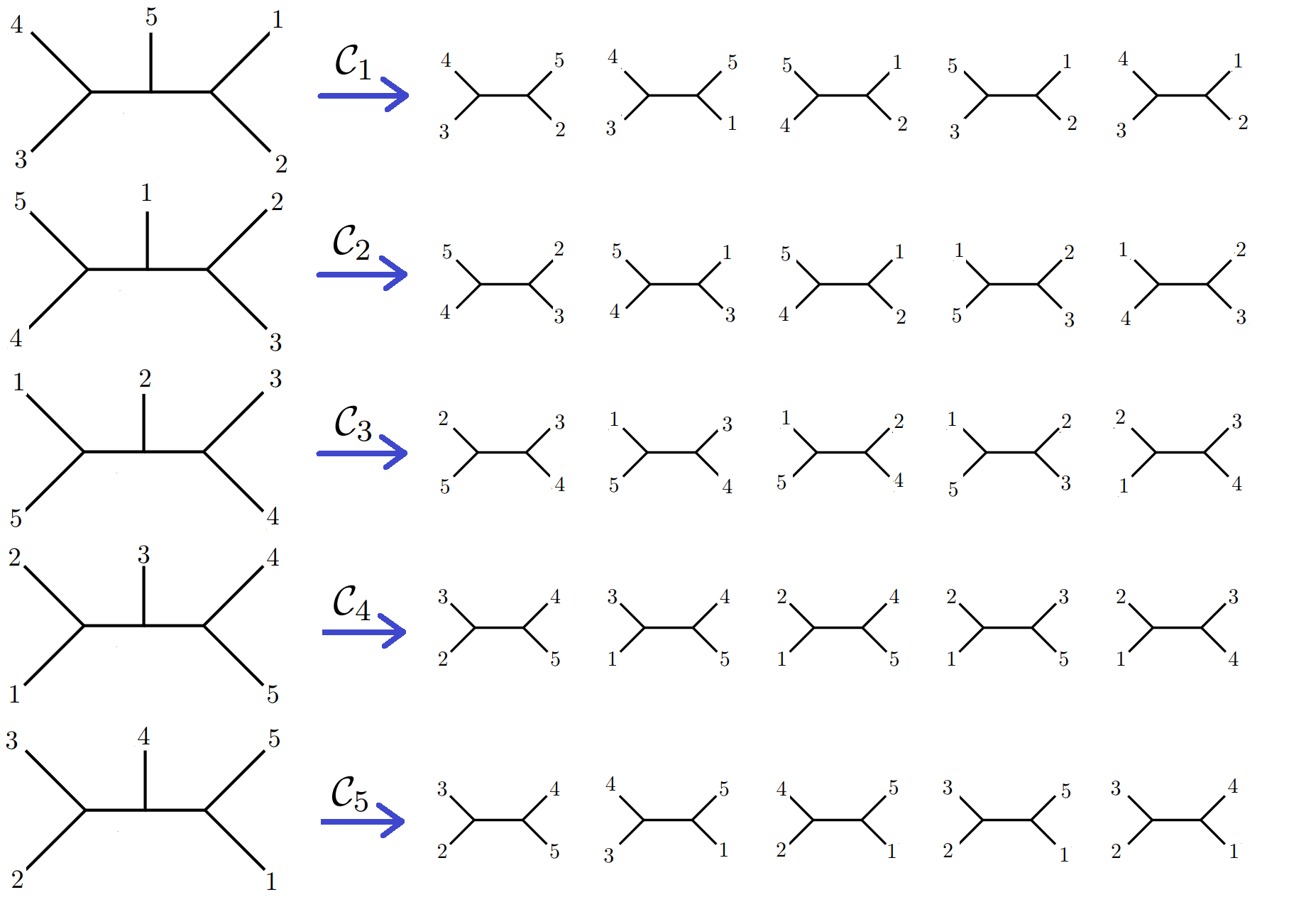}
\caption{ The five $k=2$ planar Feynman diagrams and their corresponding collections in $(k,n)=(3,5)$. \label{35col}}
\end{figure}

In what follows, we will adopt the notation $T[ab|cd]$ for a 4-point Feynman diagram with $a,b$  and $c,d$ sharing a vertex, i.e.
\be 
T[ab|cd] :=\raisebox{-.8cm}{
\tikz[scale=.6]{\draw[thick] (-1,1) node [left]{$b$}--(0,0)--(-1,-1) node [left]{$a$} (0,0)--(1,0)--(2,1)node [right]{$c$} (1,0)--(2,-1) node [right]{$d$};} }
\ee 
\iffalse
Thus we have 
 \be \label{perm5fads}
 \!\!\!\!\!
E_{3,5}=\{{\cal C}_1,{\cal C}_2,{\cal C}_3,{\cal C}_4,{\cal C}_5\}
=
\left\{
\begin{array}{ccccccccc}
 T[34|52] && T[23|45] && T[34|52] && T[23|45] && T[23|45] \\
 T[34|51] && T[13|45] && T[34|51] && T[13|45] && T[34|51] \\
 T[12|45] &,& T[12|45] &,& T[24|51] &,& T[12|45] &,& T[24|51] \\
 T[12|35] && T[23|51] && T[23|51] && T[12|35] && T[23|51] \\
 T[12|34] && T[23|41] && T[12|34] && T[12|34] && T[23|14] \\
\end{array}
\right\}.
\ee
\fi 
By applying cyclic permutations \eqref{cycper} on $E_{3,5}$ we get the set $E_{3,5}^{(1)}$, $E_{3,5}^{(2)}$, $\cdots$, $E_{3,5}^{(6)}$ with $E_{3,5}^{(6)}=E_{3,5}$. In the more compact notation defined above we have, for instance
 \footnote{Here, for example, one can see the cyclic permutation  $\{1\to3, 2\to4, 3\to5,4\to6,5\to1,6\to2\}$ of ${\cal C}_1$ as 
 $\{T[45|63],T[45|62],T[23|56],T[23|46],T[23|45]\}$ with the leaves $3$, $4$, $5$, $6$ and $1$ pruned, respectively, in addition to the common missing leaf $2$.  We rotate the list from right by 1 to get a planar collection ${\cal C}_{1}^{(2)}$ with the leaves $1$, $3$, $4$, $5$ and $6$ pruned, respectively.}
 \begin{align} \label{perm5}
&E_{3,5}^{(1)}=\{{\cal C}_1^{(1)},{\cal C}_2^{(1)},{\cal C}_3^{(1)},{\cal C}_4^{(1)},{\cal C}_5^{(1)}\}
=
\left\{
\begin{array}{ccccccccc}
 T[45|63] && T[34|56] && T[45|63] && T[34|56] && T[34|56] \\
 T[45|62] && T[24|56] && T[45|62] && T[24|56] && T[45|62] \\
 T[23|56] &,& T[23|56] &,& T[35|62] &,& T[23|56] &,& T[35|62] \\
 T[23|46] && T[34|62] && T[34|62] && T[23|46] && T[34|62] \\
 T[23|45] && T[34|52] && T[23|45] && T[23|45] && T[34|52] \\
\end{array}
\right\}\,,
 \end{align}
  \begin{align} \label{perm55}
E_{3,5}^{(2)}=\{{\cal C}_1^{(2)},{\cal C}_2^{(2)},{\cal C}_3^{(2)},{\cal C}_4^{(2)},{\cal C}_5^{(2)}\}
=
\left\{
\begin{array}{ccccccccc}
 T[34|56] && T[45|63] && T[34|56] && T[34|56] && T[45|63] \\
 T[14|56] && T[45|61] && T[14|56] && T[45|61] && T[45|61] \\
 T[13|56] &,& T[35|61] &,& T[13|56] &,& T[35|61] &,& T[13|56] \\
 T[34|61] && T[34|61] && T[13|46] && T[34|61] && T[13|46] \\
 T[34|51] && T[13|45] && T[13|45] && T[34|51] && T[13|45] \\
\end{array}
\right\}\,.
 \end{align}

The idea of the second bootstrap is that each column of a planar matrix is a planar collection. In other words, a planar matrix must take the form
\be\label{form6666}
{\cal M}= [ {\cal C}_{i_1}^{(1)},  {\cal C}_{i_2}^{(2)}, \cdots, {\cal C}_{i_6}^{(6)} ] \,, \quad \text{with}~ 1\leq i_1,\cdots,i_6 \leq 5\,, 
\ee 
where each element in ${\cal M}$ corresponds to a column, thus the $i$-th column belongs to the set $E_{3,5}^{(i)}$ subject to the $i$-th permutation. There are five choices for the first column, since there are five collections in $(3,5)$. However, once one of the collections is chosen, the choices for the remaining five columns get substantially reduced.  

For example, let's choose the first column of the matrix to be the first collection  ${\cal C}_{1}^{(1)}$. The symmetry of the matrix implies $T^{(1,2)}=T^{(2,1)}$, thus the first tree of the second column  $ {\cal C}_{i_2}^{(2)}$ must be the first tree of the first column, i.e. $T[45|63]$ \footnote{Recall that ${\cal C}_{1}^{(1)}= \{T[45|63],$ $T[45|62],$ $T[23|56],$ $T[23|46],$ $T[23|45] \}$.}. By looking at \eqref{perm5} and \eqref{perm55} we find that only $ {\cal C}_{2}^{(2)}$ and $ {\cal C}_{5}^{(2)}$ satisfy this requirement. 
Similarly, we select candidates from $E_{3,5}^{(i)}$ by again imposing the symmetry condition
$T^{(1,i)}=T^{(i,1)}$ now
for $i=3,4,5,6$ (see figure \ref{fig:universfade3} for a sketch).

\begin{figure}[!htb]
\centering
\begin{overpic}[scale=0.25]{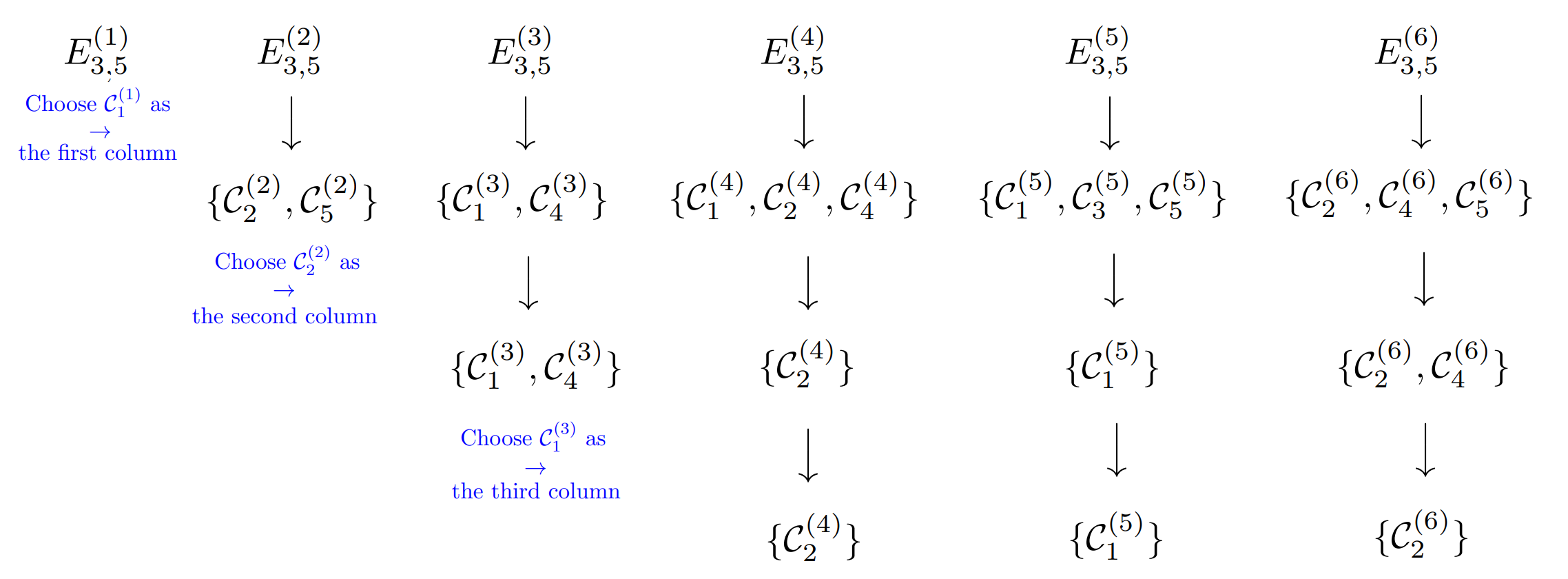}
\end{overpic}
\caption{\label{fig:universfade3} Illustration of the second combinatorial bootstrap for obtaining planar matrices of Feynman diagrams. Here we choose $ {\cal C}_{1}^{(1)}$, $ {\cal C}_{2}^{(2)}$ and $ {\cal C}_{1}^{(3)}$ as the first three columns and then get a symmetric planar matrix by filling in the remaining three columns with  $ {\cal C}_{2}^{(4)}$, $ {\cal C}_{1}^{(5)}$ and $ {\cal C}_{2}^{(6)}$. See also table \ref{46matrices22}.  }
\end{figure}
With this approach, the number of choices for the remaining 5 columns has been reduced from the naive $5^5=3\,125$ to $2\times 2\times 3\times 3\times 3= 108$. 

Therefore, we can now forget about the first column and focus on the possible 108 choices for the remaining 5 columns. Let's for instance choose ${\cal C}_{2}^{(2)}= \{T[45|63],$ $T[45|61],$ $T[35|61],$ $T[34|61],$ $T[13|45] \}$ for the second column of the matrix. Because of the symmetry condition $T^{(2,i)}=T^{(i,2)}$ in \eqref{matrix}, only one or two candidates are selected for each of the remaining four columns, see third row in figure \ref{fig:universfade3}. 

By going on with the procedure above, we end up with a planar matrix of Feynman diagrams  
\begin{align}
&\,\,[ {\cal C}_{1}^{(1)},  {\cal C}_{2}^{(2)}, {\cal C}_{1}^{(3)},{\cal C}_{2}^{(4)},  {\cal C}_{1}^{(5)},  {\cal C}_{2}^{(6)} ] 
\\
=& 
\left[
\begin{array}{ccccccccccc}
 \emptyset && T[45|63] && T[45|62] && T[23|56] && T[23|46] && T[23|45] \\
 T[45|63] && \emptyset && T[45|61] && T[35|61] && T[34|61] && T[13|45] \\
 T[45|62] && T[45|61] && \emptyset && T[25|61] && T[24|61] && T[12|45] \\
 T[23|56] && T[35|61] && T[25|61] && \emptyset && T[23|61] && T[23|51] \\
 T[23|46] && T[34|61] && T[24|61] && T[23|61] && \emptyset && T[23|41] \\
 T[23|45] && T[13|45] && T[12|45] && T[23|51] && T[23|41] && \emptyset \\
\end{array}
\right]\,,
\nonumber
\end{align}
which happens to be ${\cal M}_1$ in table \ref{46matrices22} on the next page and is also the example shown in figure \ref{universe3}.

Had we chosen ${\cal C}_{5}^{(2)}$ for the second column instead of ${\cal C}_{2}^{(2)}$, we would have found another two planar matrices using the same procedure, which correspond to ${\cal M}_2$ and ${\cal M}_3$ in table \ref{46matrices22}. Hence, we find a total of 3 planar matrices for the initial choice ${\cal C}_{1}^{(1)}$.  

Likewise, one finds 3, 2, 4 and 2 planar matrices for the initial choices ${\cal C}_{2}^{(1)}$, ${\cal C}_{3}^{(1)}$, ${\cal C}_{4}^{(1)}$ and ${\cal C}_{5}^{(1)}$, respectively, thus giving $14$ planar matrices in total. One can check that all these 14 matrices satisfy the compatibility conditions \eqref{comp4}. Therefore, all of them contribute to the biadjoint amplitude in $k=3$. 

In table \ref{46matrices22} we present all $14$ planar matrices of Feynman diagrams in $(4,6)$, explicitly showing the corresponding collections in each column.
\begin{table}[!htb]  
  \centering
  \begin{tabular}{ |p{1.5cm}|p{3.9cm}||p{1.5cm}|p{3.9cm}|}
  \cline{1-4}
    \textbf{Matrix} & \hspace{9mm}{ 
    \textbf{Collections}
    } & \textbf{Matrix} & \hspace{9mm}\textbf{Collections}\\
 \hline
 \hspace{4mm}${\cal M}_1$   & \hspace{3.5mm}$[{\cal C}_1, {\cal C}_2, {\cal C}_1, {\cal C}_2, {\cal C}_1, {\cal C}_2]$ & \hspace{4mm}${\cal M}_8$ & \hspace{3.5mm}$[{\cal C}_3, {\cal C}_2, {\cal C}_1, {\cal C}_5, {\cal C}_4, {\cal C}_4]$\\
  \hline
\hspace{4mm}${\cal M}_2$ & \hspace{3.5mm}$[{\cal C}_1, {\cal C}_5, {\cal C}_4, {\cal C}_4, {\cal C}_3, {\cal C}_2]$ & \hspace{4mm}${\cal M}_9$ & \hspace{3.5mm}$[{\cal C}_4, {\cal C}_4, {\cal C}_3, {\cal C}_2, {\cal C}_1, {\cal C}_5]$\\
\hline
 \hspace{4mm}${\cal M}_3$   & \hspace{3.5mm}$[{\cal C}_1, {\cal C}_5, {\cal C}_4, {\cal C}_1, {\cal C}_5, {\cal C}_4]$ & \hspace{4mm}${\cal M}_{10}$ & \hspace{3.5mm}$[{\cal C}_4, {\cal C}_1, {\cal C}_5, {\cal C}_4, {\cal C}_1, {\cal C}_5]$ \\
  \hline
\hspace{4mm}${\cal M}_4$ & \hspace{3.5mm}$[{\cal C}_2, {\cal C}_4, {\cal C}_3, {\cal C}_2, {\cal C}_4, {\cal C}_3]$ & \hspace{4mm}${\cal M}_{11}$ & \hspace{3.5mm}$[{\cal C}_4, {\cal C}_3, {\cal C}_2, {\cal C}_4, {\cal C}_3, {\cal C}_2]$\\
\hline
 \hspace{4mm}${\cal M}_5$   & \hspace{3.5mm}$[{\cal C}_2, {\cal C}_1, {\cal C}_5, {\cal C}_4, {\cal C}_4, {\cal C}_3]$ & \hspace{4mm}${\cal M}_{12}$ & \hspace{3.5mm}$[{\cal C}_4, {\cal C}_3, {\cal C}_2, {\cal C}_1, {\cal C}_5, {\cal C}_4]$\\
  \hline
\hspace{4mm}${\cal M}_6$ & \hspace{3.5mm}$[{\cal C}_2, {\cal C}_1, {\cal C}_2, {\cal C}_1, {\cal C}_2, {\cal C}_1]$ & \hspace{4mm}${\cal M}_{13}$ & \hspace{3.5mm}$[{\cal C}_5, {\cal C}_4, {\cal C}_4, {\cal C}_3, {\cal C}_2, {\cal C}_1]$\\
\hline
\hspace{4mm}${\cal M}_7$ & \hspace{3.5mm}$[{\cal C}_3, {\cal C}_2, {\cal C}_4, {\cal C}_3, {\cal C}_2, {\cal C}_4]$ & \hspace{4mm}${\cal M}_{14}$ & \hspace{3.5mm}$[{\cal C}_5, {\cal C}_4, {\cal C}_1, {\cal C}_5, {\cal C}_4, {\cal C}_1]$\\
\hline
  \end{tabular}
  \caption{ \label{46matrices22} Planar matrices of Feynman diagrams in $(4,6)$. Here we 
  abbreviate $[ {\cal C}_{i_1}^{(1)},  {\cal C}_{i_2}^{(2)}, \cdots, {\cal C}_{i_6}^{(6)} ] $ as $[ {\cal C}_{i_1},  {\cal C}_{i_2}, \cdots, {\cal C}_{i_6} ] $
 since the superscripts can be inferred from the position of ${\cal C}_i$ in the brackets.}
\end{table}

\subsection{A More Interesting Example: From $(3,6)$ to $(4,7)$} \label{sec33}

Now we comment on another example, in this case on how to obtain planar matrices of Feynman diagrams for $(4,7)$ using the second bootstrap again. The starting point are the 48 planar collections of (3,6), i.e.
%shown in appendix C
$E_{3,6}=\{ {\cal C}_1,{\cal C}_2,\cdots, {\cal C}_{48} \}$,  which can be obtained from the first bootstrap.  The cyclic permutations \eqref{cycper}  give the set $E_{3,6}^{(1)}$, $E_{3,6}^{(2)}$, $\cdots$, $E_{3,6}^{(7)}$ with $E_{3,6}^{(7)}=E_{3,6}$.  Then
a planar matrix must take the form
\be\label{48448}
{\cal M}= [ {\cal C}_{i_1}^{(1)},  {\cal C}_{i_2}^{(2)}, \cdots, {\cal C}_{i_7}^{(7)} ] \,, \quad \text{with}~ 1\leq i_1,\cdots,i_7 \leq 48\,,
\ee 
where the $i$-th column belongs to the set $E_{3,6}^{(i)}$. Now we have 48 choices for the first column. Once again, we repeat the same procedure as before but now for 7 columns, and we get 693 planar matrices. One can check that all these 693 symmetric matrices satisfy the compatibility conditions \eqref{comp4}. Therefore, all of them  are  planar matrices and contribute to the biadjoint amplitude in $k=4$. 

After summing over every choice of the first column as well as every possible choice for the remaining columns allowed by the candidates at each step, we get  $693$ symmetric matrices in total,  which are much more than the $42$ initial planar matrices for $(4,7)$ used in the pruning-mutation procedure of section \ref{sec2}. There are $693$ planar collections in $(3,7)$ as well and how they are dual to  $693$ planar matrices is explained in section \ref{sec6}.

The ordering of collections in $E_{3,6}$ is not relevant as long as its cyclic permutations $E_{3,6}^{(1)},\cdots,E_{3,6}^{(7)}$ change covariantly. For the readers' convenience, we borrow Table 1 from \cite{Borges:2019csl} containing all $48$ collections and place it as table \ref{tb:n6coll} in appendix \ref{apB}. We adopt the same ordering notation as in \cite{Borges:2019csl} so that we can  present more details of the second bootstrap. 

A collection in table \ref{tb:n6coll} is given by 6 trees characterized by 6 numbers. For example, the first collection ${\cal C}_1$  expressed by $[4,4,4,3,3,3]$ means the collection given in figure \ref{fig:universe2}, where the ``middle leaves" are $4$, $4$, $4$, $3$, $3$ and $3$ respectively. Its cyclic permutations give ${\cal C}_1^{(1)},{\cal C}_1^{(2)},\cdots,{\cal C}_1^{(7)}$ with ${\cal C}_1^{(7)}={\cal C}_1$, which act as the first element of   $E_{3,6}^{(1)},E_{3,6}^{(2)},\cdots,E_{3,6}^{(7)}$ respectively.

If we choose ${\cal C}_1^{(1)}$ as the first column, it happens that from each $E_{3,6}^{(2)},\cdots,E_{3,6}^{(7)}$ there are 14 collections satisfying 
the symmetry requirement $T^{(1,i)}=T^{(i,1)}$. For example, for the second and third column,   their 14 possible choices of collections are 
\begin{align} 
&{\cal C}_1^{(2)},\,{\cal C}_{15}^{(2)},\,{\cal C}_{19}^{(2)},\,{\cal C}_{24}^{(2)},\,{\cal C}_{26}^{(2)},\,{\cal C}_{34}^{(2)},\,{\cal C}_{39}^{(2)},\,{\cal C}_{42}^{(2)},\,{\cal C}_{43}^{(2)},\,{\cal C}_{44}^{(2)},\,{\cal C}_{45}^{(2)},\,{\cal C}_{46}^{(2)},\,{\cal C}_{47}^{(2)},\,{\cal C}_{48}^{(2)}\,,
\label{693uu}
\\
&{\cal C}_3^{(3)},\,{\cal C}_{7}^{(3)},\,{\cal C}_{10}^{(3)},\,{\cal C}_{14}^{(3)},\,{\cal C}_{17}^{(3)},\,{\cal C}_{21}^{(3)},\,{\cal C}_{28}^{(3)},\,{\cal C}_{31}^{(3)},\,{\cal C}_{32}^{(3)},\,{\cal C}_{33}^{(3)},\,{\cal C}_{36}^{(3)},\,{\cal C}_{38}^{(3)},\,{\cal C}_{41}^{(3)},\,{\cal C}_{48}^{(3)}\,.
\label{693uu22}
\end{align} 
We see that the naive number of choices for the remaining 6 columns reduces from $48^6\sim 1\times 10^{10}$ down to $14^6\sim 8\times 10^{6}$.

Now  we can forget the first column and focus on the  $14^6$ candidates for the remaining 6 columns. If we choose ${\cal C}_1^{(2)}$ from 
\eqref{693uu} as the second column, we find 
only one collection ${\cal C}_{48}^{(3)}$ from \eqref{693uu22} that satisfies the requirement $T^{(2,3)}=T^{(3,2)}$. Similarly, we find that there is only one collection from 14 candidates for the remaining 4 columns satisfying the requirement $T^{(2,i)}=T^{(i,2)}$ as well for $i=4,5,6,7$. This time we see that the naive number of choices for the remaining five columns dramatically reduces from $14^5\sim 5\times 10^{5}$ to $1$. Hence the only choice that makes up a planar matrix of the form \eqref{48448} is 
 \be \label{315315}
 [ {\cal C}_{1}^{(1)},  {\cal C}_{1}^{(2)}, {\cal C}_{48}^{(3)}, {\cal C}_{41}^{(4)}, {\cal C}_{27}^{(5)}, {\cal C}_{18}^{(6)},  {\cal C}_{8}^{(7)} ]\,.
 \ee 

Had we chosen the remaining collections  ${\cal C}_{15}^{(2)},\cdots,  {\cal C}_{47}^{(2)} $ or ${\cal C}_{48}^{(2)}$ in \eqref{693uu} as the second column instead of ${\cal C}_{1}^{(2)}$, we would have found $  1$, $ 2$, $1$, $2$, $1$, $2$, $2$, $2$, $1$, $2$, $3$, $3$ and $9$ planar matrices respectively.  Thus there are 32 planar matrices in total with ${\cal C}_{1}^{(1)}$ as the first column.

  Similarly, we can get all of the planar matrices with  ${\cal C}_{2}^{(1)},\cdots, {\cal C}_{47}^{(1)}$ or ${\cal C}_{48}^{(1)}$  as the first column of the matrix.  By adding them up, including the 32 ones for ${\cal C}_{1}^{(1)}$, we obtain all the 693 planar matrices in $(4,7)$.

\section{Main Results} \label{sec4}

The main applications of the techniques introduced in this work are the computation of all the planar collections of Feynman diagrams for the cases $(3,6)$, $(3,7)$, $(3,8)$ and $(3,9)$. This is done using the first kind of combinatorial bootstrap. We have also computed all the planar matrices of Feynman diagrams for $(4,7)$, $(4,8)$ and $(4,9)$. 
% \yz 
% All of these results are saved as  ancillary files in a folder titled {\tt anc} and we provide a  
% {\textsc{Mathematica}} 
% notebook named 
% {\tt {Arrays\_of\_FDs.nb}}
% to read them automatically.
% \yz
%
In this section we try to give a self-contained presentation of these results 
in the form of {\tt *.m} files 
% \footnote{Note that {\tt .m} files cannot only be opened by \textsc{Mathematica} but also by any \textsc{TextEdit}. When there is a file with name in the form {\tt *c.m}, it is compressed to reduce its size and can be uncompressed by  using the command ``Uncompress@Import@*c.m'' in any \textsc{Mathematica} notebook with correct directory.
% }
and 
a \textsc{Mathematica} notebook named {\tt Arrays\_of\_FDs.nb} to read them.

Note that {\tt .m} files can not only be opened by \textsc{Mathematica} but also by any \textsc{TextEdit}. When there is a file with name in the form {\tt *c.m}, it is compressed to reduce its size and can be uncompressed by  using the command ``Uncompress@Import@*c.m'' in any \textsc{Mathematica} notebook with correct directory.

% \yz 
% All .m files with small size are present as there while those with huge size are compressed by \textsc{Mathematica} 
% \yz

% The \textsc{Mathematica} notebook as an auxiliary file along with  the data is contained in text files available at \url{https://www.dropbox.com/sh/w5i3vhig1qm1r0f/AAAuF-vtRxUCFRj5BTAqiF9wa?dl=0}. 

\subsection{Planar Collections of Feynman Diagrams} \label{sec41}

We have placed the results of all planar collections of Feynman diagrams for $(3,6)$, $(3,7)$, $(3,8)$, and $(3,9)$ as {\tt .m} files.  The files {\tt col36.m} and {\tt col37.m} are human readable and {\tt col38c.m} and {\tt col39c.m}  will become readable after being uncompressed by \textsc{Mathematica}.

Let us illustrate the content of the files. For example, in the file {\tt col36.m}, there are 
all the $14$ planar collections of Feynman diagrams for $(3,6)$. Each of the $14$ planar collections is a set of six $5$-point Feynman diagrams. The way we choose to store the information is better explained with an example. The first collection presented in the file reads
\begin{align}\label{example1}
{\rm c}[\,\,&{\rm FD}[s[2,3],s[2,3,4], 1],\,{\rm FD}[s[1,3],s[1,3,4], 2],\,
{\rm FD}[s[1,2],s[1,2,4],3],\,
\nonumber
\\
&{\rm FD}[s[1,2],s[1,2,3],4],\,
{\rm FD}[s[1,2],s[1,2,3],5],\,
{\rm FD}[s[1,2],s[1,2,3], 6]\,\,]
    \,.
\end{align}
Here we have made use of the fact that a $5$-point Feynman diagram can be completely characterized by its two poles. Note that we have to assume that each Feynman diagram in the collection comes endowed with its valid kinematic data, i.e. for the $i^{\rm th}$-tree one uses  $s^{(i)}_{jk}$ with $j,k\in \{1,2,3,4,5,6\}\setminus i$ and satisfying momentum conservation in only the five particles present. In section \ref{sec6} we show that these $n$ copies of kinematic spaces are more than just a convenience and that each tree is indeed a fully fledged Feynman diagram. 

Let us continue with the example in \eqref{example1}. As mentioned above, the six 5-point Feynman diagrams have leaves $1,2,\cdots, 6$ pruned respectively.
Thus ${\rm FD}[s[2,3],s[2,3,4], 1]$ means a Feynman diagram with the poles $s_{2,3}^{(1)}$, $s_{2,3,4}^{(1)}=s^{(1)}_{5,6}$ and with the leaf 1 pruned. This collection happens to be the example shown in figure \ref{fig:universe2}.  

The contribution of a given planar collection of Feynman diagrams to a $k=3$ biadjoint amplitude can be computed by integrating over the space of compatible metrics, e.g. using the formula \eqref{intew}. For the cases $(3,6)$ and $(3,7)$ this is easily done and the results are in agreement with previous computations \cite{Cachazo:2019apa,Cachazo:2019ngv,Drummond:2019qjk}. However, we find that a straightforward application of such a method to $(3,8)$ and $(3,9)$ is not practical with modest computing resources. This is why we developed much more efficient but equivalent algorithms for computing such contributions to the amplitudes. In this section we simply present the data and postpone the explanation of the algorithm to the next section where we discuss evaluations as computing volumes.

For convenience,  we created a \textsc{Mathematica} notebook named {\tt Arrays\_of\_FDs.nb} that can read all the .m files results automatically.  Let us illustrate the use of {\tt Arrays\_of\_FDs.nb} with an example. The output of the command {\tt col}[3,8] gives all $13\, 612$ planar collections for $(3,8)$ taken from {\tt col38c.m}.

We postpone the presentation of the integrated amplitude $k=3$ biadjoint amplitude $m^{(3)}_8(\mathbb{I},\mathbb{I})$, using a more efficient method, to the next section.

\subsection{Planar Matrices of Feynman Diagrams} \label{sec42}
 
We express the planar matrices 
for $(4,7)$, $(4,8)$, and $(4,9)$ as a series of collections we already save and their cyclic permutations. More explicitly, recall that the second kind of bootstrap is based on the fact that each column of a planar matrix for $n$-points must be one of the planar collections for $(n-1)$-points with the labels chosen appropriately.  

For example, in the file mat47.m there are 
$693$ sets for $(4,7)$. The first one reads,
\begin{align}\label{example12}
\{1, 1, 48, 41, 27, 18, 8\}\,.
\end{align}
This notation might seem cryptic at first but it is actually both very efficient and simple. In order to gain familiarity with the notation note that this is exactly the  matrix presented in \eqref{315315} but it is given with a slightly less compact notation.

In practice, this means that we can get the planar matrix by picking out the $1^{\rm st},17^{\rm th},\cdots,$ $ 3^{\rm th}$ collections in col36.m and shifting their particle labels in a cyclic ordering $(1,2,3,4,5,6,7)$  by $1,2,3,4,5,6,7$ respectively.

For the user's convenience, we introduced a command ${\tt matrix}[k\_,n\_][set\_]$ in the \textsc{Mathematica} notebook {\tt Arrays\_of\_FDs.nb} to expand the compact notation and produce the explicit planar matrices. In addition, the command ${\tt mat}[k\_,n\_]$ returns the complete set of all planar matrices automatically. For example, the output of {\tt matrix}[4,7][\{1, 1, 48, 41, 27, 18, 8\}] gives the explicit expression of the planar matrix shown in \eqref{315315}.

All planar matrices for $(4,8)$ and $(4,9)$ collected in {\tt mat48c.m} and  {\tt mat49seedc.m} respectively which can be read by \textsc{Mathematica}. 
So one can use {\tt mat}[4,9] to produce all 30\, 659 \,424 (4,9) planar matrices.
The output of 
\be 
{\tt matrix}[4,9][\{98, 280, 5154, 7773, 9509, 11334, 10639, 9515, 5082\}]
\ee 
gives an explicit expression of a planar matrix for $(4,9)$, whose columns are from the $98^{\rm th}$, $280^{\rm th}$, $\cdots$, and $5082^{\rm nd}$ collections  of  $(3,8)$ presented in the file {\tt col38c.m} respectively.

With all planar collections of Feynman diagrams for (3,8) and all planar matrices for (4,8) and (4,9), one can in principle to produce any $(3,8),(4,8),(4,9)$ partial amplitudes with two arbitrary orderings according to \eqref{intew} and \eqref{intFTM}. With advanced techniques explained in the next section, their amplitudes can be computed much faster.  We provide a notebook file   named {\tt 38\_48\_amplitudes.nb} where one can set any  kinematics satisfying \eqref{1333} or \eqref{condS}  as input  to produce numeric amplitudes $m^{(3)}_8(\mathbb{I}|\mathbb{I})$ and $m^{(4)}_8(\mathbb{I}|\mathbb{I})$ in reasonable time. \footnote{It takes around one hour for $(4,8)$ on a laptop. Attentive readers are welcome to  modify the codes in the notebook file to get any analytic partial amplitudes for (3,8) and (4,8) but the computation would take a longer time.} 
We also provide another notebook file  named {\tt 49\_amplitudes.nb} whose computations have been finished and which contains several numeric diagonal amplitude  $m^{(4)}_9(\mathbb{I}|\mathbb{I})$ and several analytic off-diagonal amplitudes  $m^{(4)}_9(\mathbb{I}|\beta)$.  More details will be explained in the next section.

% The output of {\tt mat}[4,8] gives all of the $90\,608$ planar matrices for $(4,8)$. 
% % The contribution to an amplitude of each of the planar matrices in $(4,8)$ are saved in the file ampList48.txt.
% Similarly to the $(3,8)$ case, we present a notebook with the full amplitude $m^{(4)}_8(\mathbb{I}|\mathbb{I})$ in the next section. 

% Several $(4,9)$ amplitudes are also included in another auxiliary file named {\tt 49\_amplitudes.nb} which will  also be explained  later.

\section{Evaluation: Computing Volumes from Planar Collections and Matrices} \label{sec5}

In this section we introduce a geometric characterization of each collection as a facet of the full polytope. Motivated by the picture provided in \cite{Drummond:2019qjk}, we then realize the amplitude \eqref{intFTM} as the volume of the corresponding geometry. The characterization speeds up the evaluation, especially for the cases $(3,8)$,  $(4,8)$ and $(4,9)$, as it can be implemented and automatized via the software PolyMake. The full amplitudes are provided in ancillary files. 

Let us describe the procedure for the collection of figure \ref{bipcol}, where the compatibility equations fix:
 
\begin{eqnarray}
u =& x+w-z \\
v =& y+w-z
\end{eqnarray}
\begin{figure}[!htb]
\centering
\includegraphics[width=150mm]{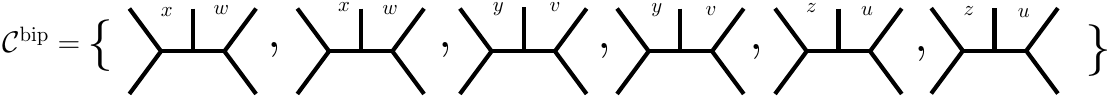}
\caption{Bipyramidal collection for $(3,6)$}
\label{bipcol}
\end{figure}
As explained in \cite{Cachazo:2019ngv,Borges:2019csl} this corresponds to a bipyramid facet in the tropical Grassmannian $(3,6)$. To make this identification precise, here we note that the six constraints on the collection metric, namely $x,y,z,w,u,v>0$ can be written as
\[
F\cdot Z_{i}>0
\]
where we define $F:=(x,y,z,w) \in \mathbb{R}^4$ and 
\begin{align}
Z_{1} & =(1,0,0,0)\,,\,Z_{2}=(0,1,0,0)\,,\,Z_{3}=(0,0,1,0)\,,\nonumber \\
Z_{4} & =(0,0,0,1)\,,\,Z_{5}=(0,1,-1,1)\,,\,Z_{6}=(1,0,-1,1)\,.\label{eq:6z}
\end{align}
Each $Z_{i}$ corresponds to a plane that passes through the origin. As will be further explored in \cite{Guevara:2020lek}, such planes correspond to degenerated collections and hence to boundaries of our geometry.
Indeed, the inequalities define a cone in $\mathbb{R}^{4}$

\[
\Delta=\{F\in\mathbb{R}^{4}|\,F\cdot Z_{i}>0\,,i=1,\ldots,6\}\,.
\]

A bounded three dimensional region is obtained by intersecting $\Delta$
with an affine plane, e.g. $F\cdot(1,1,1,1)=x+y+z+w=1$. This allows
to draw a three dimensional bipyramid as in figure \ref{bipbip}, where we label the faces by the corresponding $Z_i$. For our purposes
there is a preferred affine plane $T$ given by

\begin{figure}[!htb]
\centering
\includegraphics[width=120mm]{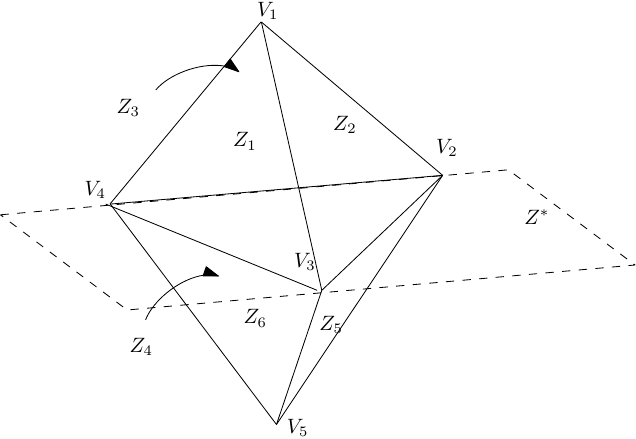}
\caption{Bipyramid projected into a three dimensional slice, for instance by imposing $F\cdot T=1$. Three or more planes $Z_i$ intersect at vertices $V_i$. We have also depicted the auxiliary plane $Z^*$ passing through $V_2,V_3,V_4$.}
\label{bipbip}
\end{figure}
\begin{equation}\label{eq:FT1}
\mathcal{F}(\mathcal{C}^{\rm bip})=F\cdot\underbrace{(R,t_{1234},t_{3456},t_{5612}-R)}_{T}=1\,.
\end{equation}
\begin{figure}[!htb]
\centering
\includegraphics[width=150mm]{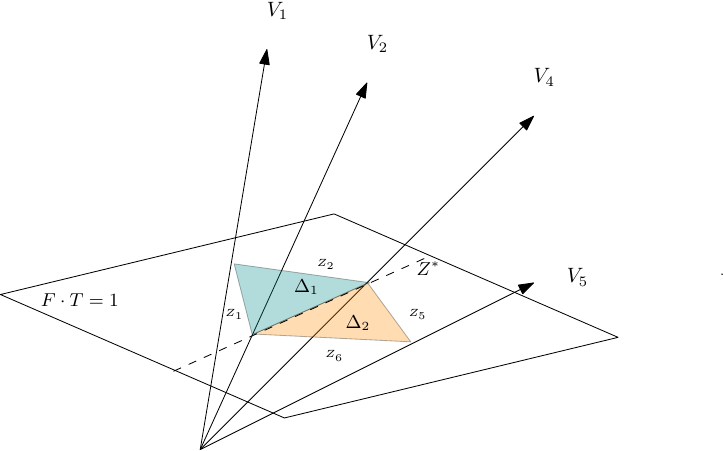}
\caption{A cartoon for the bipyramid as a cone in $\mathbb{R}^4$, bounded by planes $Z_i$. The vertices $V_i$ correspond to rays. }
\label{bipbip2}
\end{figure}

We now show that the volume of the corresponding three dimensional
bipyramid is precisely the amplitude associated to such collection,
as defined by the Laplace integral formula $\eqref{intew}$
\[
\int_{\Delta}d^{4}F\,e^{-\mathcal{F}(\mathcal{C}^{\rm bip})}\, .% = \int_{\Delta}dx\,dy\,dz\,dw\,\exp[-( x\,R+y\,t_{1234}+z\,t_{3456}+w\,(t_{5612}-R))]
\]
The statement is a particular case of a more general fact known as the Duistermaat-Heckman formula\footnote{ which has recently been applied in \cite{Frost:2018djd} for the context of string integrals.}.
To derive it, we note that the bipyramid $\Delta$ (projected into
$F\cdot T=1$) can be triangulated by two simplices, $\Delta=\Delta_{1}\cup\Delta_{2}$
by inserting the auxiliary plane, also depicted in figure \ref{bipbip}
\[
F\cdot Z^{*}:=y-z=0
\]
corresponding to the base of the bipyramid. In practice, $Z^{*}=(0,1,-1,0)$
can be easily found from the input (\ref{eq:6z}) using the {\tt TRIANGULATION}
function of \textsc{PolyMake}. The two simplices are defined by the cones
\begin{align*}
\Delta_{1} & =\{F\in\mathbb{R}^{4}|\,F\cdot Z_{1}>0,F\cdot Z_{3}>0,F\cdot Z_{5}>0,F\cdot Z^{*}>0\}\\
\Delta_{2} & =\{F\in\mathbb{R}^{4}|\,F\cdot Z_{2}>0,F\cdot Z_{4}>0,F\cdot Z_{6}>0,F\cdot Z^{*}<0\}
\end{align*}
after their projection to three dimensions via $F\cdot T=1$, see figure \ref{bipbip2}. Now,
\[
\int_{\Delta}d^{4}F\,e^{-\mathcal{F}(\mathcal{C}^{\rm bip})}=\int_{\Delta_{1}}d^{4}Fe^{-F\cdot T}+\int_{\Delta_{2}}d^{4}Fe^{-F\cdot T}
\]
and it suffices to show that each integral computes the volume of
the corresponding simplex. For this we introduce the set of vertices
for each simplex. For $\Delta_{1}$ we have the vertices (rays) $\{V_{1},V_{2},V_{3},V_{4}\}$
defined by
\begin{align}\label{dualsys}
V_{1}\cdot Z_{1} & =V_{1}\cdot Z_{2}=V_{1}\cdot Z_{3}=0\\
V_{2}\cdot Z_{2} & =V_{2}\cdot Z_{3}=V_{2}\cdot Z^{*}=0\\
V_{3}\cdot Z_{1} & =V_{3}\cdot Z_{2}=V_{3}\cdot Z^{*}=0\\
V_{4}\cdot Z_{3} & =V_{4}\cdot Z_{1}=V_{4}\cdot Z^{*}=0.
\end{align}
We then have
\[
F\in\Delta_{1}\iff F=\sum_{i=1}^{4}\beta^{i}V_{i}\,,\quad\beta^{i}>0
\]
and simple algebra leads to 
\begin{equation}
\int_{\Delta_{1}}d^{4}Fe^{-F\cdot T}=\frac{| \langle V_{1},V_{2},V_{3},V_{4}\rangle |}{(V_{1}\cdot T)(V_{2}\cdot T)(V_{3}\cdot T)(V_{4}\cdot T)}
,\label{eq:vt}
\end{equation}
where $\langle V_{1},V_{2},V_{3},V_{4}\rangle = \epsilon^{IJKL} V_{1I}V_{2J}V_{3K}V_{4L}$. The absolute value is conventional as it corresponds to a choice of orientation.
This is precisely the volume spanned by the vectors $\hat{V}_{i}=\frac{V_{i}}{V_{i}\cdot T}$
that also belong to the plane $F\cdot T=1$, i.e. they satisfy $\hat{V}_{i}\cdot T=1$. This thus proves that the
amplitude is given by the volume of the cone $\Delta$ projected into
$F\cdot T=1$.

As will be further explored in \cite{Guevara:2020lek}, we note that the vertices $V_{i}$ are nothing but the poles emerging in the associated amplitude. This realizes the geometrical intuition provided in \cite{Cachazo:2019apa} that the bipyramid is formed by the poles $\{R,\tilde{R},t_{1234},t_{3456},t_{5612}\}$, where $R,\tilde{R}$ correspond to the apices. Indeed, recalling the definitions
\begin{eqnarray}
t_{abcd}:=&{\bf s}_{abc}+{\bf s}_{abd}+{\bf s}_{acd}+{\bf s}_{bcd}\\
R_{ab,cd,ef}:=& t_{abcd}+{\bf s}_{cde}+{\bf s}_{cdf} \,,
\end{eqnarray}
with $R=R_{12,34,56}\,,\,\tilde{R}=R_{34,12,56}$ we find
\[
\frac{|\langle V_{1},V_{2},V_{3},V_{4}\rangle |}{(V_{1}\cdot T)(V_{2}\cdot T)(V_{3}\cdot T)(V_{4}\cdot T)}=\frac{1}{t_{1234}t_{3456}t_{5612}\tilde{R}}
\]
which corresponds to vertices of the upper half of the bipyramid in
figure \ref{bipbip}. Finally, one can check that the vertices of $\Delta_{2}$ are
given by $\{V_{2},V_{3},V_{4},V_{5}\}$ where the new ray is defined
by $V_{5}\cdot Z_{4}=V_{5}\cdot Z_{5}=V_{5}\cdot Z_{6}=0$. This means
that
\begin{align*}
\int_{\Delta_{2}}d^{4}Fe^{-F\cdot T} & =\frac{|\langle V_{1},V_{2},V_{3},V_{5}\rangle |}{(V_{1}\cdot T)(V_{2}\cdot T)(V_{3}\cdot T)(V_{5}\cdot T)}\\
 & =\frac{1}{t_{1234}t_{3456}t_{5612}R}
\end{align*}
in agreement with e.g. \cite{Cachazo:2019apa}. We  put the details on how to realize these calculations in \textsc{PolyMake} in appendix \ref{polymake}.

For the cases $(3,8)$ and $(4,8)$, this method turns out to be necessary. We have used \textsc{PolyMake} to triangulate the cone $\Delta$ for every one of their planar collections or matrices and stored the results in the ancillary files. The remaining procedure to produce the full numeric rational integrated amplitudes $m^{(k)}_8(\mathbb{I},\mathbb{I})$ for any given kinematics is  implemented in the \textsc{Mathematica} notebook {\tt 38\_48\_amplitudes.nb}. The (4,8) amplitude $m^{(4)}_8(\mathbb{I},\mathbb{I})$ obtained this way was later found to coincide with that of \cite{He:2020ray} using Arkani-Hamed-Bai-He-Yan construction \cite{Arkani-Hamed:2017tmz} in the context of Grassmannian stringy integrals \cite{Arkani-Hamed:2019mrd}, which is a strong consistency check for both sides.  

% we provide the \textsc{Mathematica} notebook {\tt 38\_48\_amplitudes.nb} implementing the above formulae and providing the full amplitudes $m^{(k)}_n(\mathbb{I},\mathbb{I})$.

% In order to run the notebook, we need the files {\tt facets38.txt, vertices38.txt, facets48.txt, vertices48.txt} generated by the above \textsc{PolyMake} script. The latter can be found in the folder

% \url{https://www.dropbox.com/sh/w5i3vhig1qm1r0f/AAAuF-vtRxUCFRj5BTAqiF9wa?dl=0}
% together with {\tt ReplacementRules38.txt} and {\tt ReplacementRules48.txt}, which contain the vector $T$ (defined in \eqref{eq:FT1}) for each collection.

We continued to apply this method for the case $(4,9)$,
which has more than $3\times 10^7$ planar matrices. The number of simplices obtained by the
{\tt TRIANGULATION}
function of \textsc{PolyMake} is even much bigger:
$4\,797\,131\,092$. 
The files containing information of vertices,  facets as well as replacement rules  are too large to attach,  so we just include 4 numeric results of the full amplitudes $m^{(4)}_9(\mathbb{I},\mathbb{I})$ for 4 sets of given kinematics data and 5 analytic results of off-diagonal  amplitudes $m^{(4)}_9(\mathbb{I},\beta)$   in an auxiliary  \textsc{Mathematica} notebook {\tt 49\_amplitudes.nb}.

We close this section by emphasizing that the vertices $V_{i}$ are
dual to the planes $Z_{i}$, which establishes a connection with the dual
polytope. Explicitly, from the incidence relations \eqref{dualsys} we can use $V_{1}^{I}=\epsilon^{IJKL}Z_{1J}Z_{2K}Z_{3L}$
(and analogously for all vertices) to rewrite expression (\ref{eq:vt}) as
\begin{equation}
\int_{\Delta_{1}}d^{4}Fe^{-F\cdot T}=\frac{\langle Z_{1},Z_{2},Z_{3},Z^{*}\rangle^{3}}{\langle T,Z_{2},Z_{3},Z^{*}\rangle\langle Z_{1},T,Z_{3},Z^{*}\rangle\langle Z_{1},Z_{2},T,Z^{*}\rangle\langle Z_{1},Z_{2},Z_{3},T\rangle} \,.
\end{equation}
This is the canonical form of the dual simplex \cite{Arkani-Hamed:2017tmz}, spanned
by the rays $\{Z_{1},Z_{3},Z_{5},Z^{*}\}$. This perspective has been
explored in detail in \cite{Arkani-Hamed:2017mur}.

\section{Higher $k$ or Planar Arrays of Feynman Diagrams and Duality} \label{sec6}

Planar collections can be thought of as one-dimensional arrays while planar matrices as two-dimensional arrays of Feynman diagrams satisfying certain conditions. It is natural to propose that the computation of generalized biadjoint amplitudes for any $(k,n)$ can be done using $k-2$ dimensional arrays of Feynman diagrams. 

\begin{defn}
A {\it planar array of Feynman diagrams} is a $(k-2)$-dimensional array ${\cal A}$ with dimensions of size $n$. The array has as component ${\cal A}_{i_1,i_2,\ldots ,i_{k-2}}$ a metric tree with leaves in the set $\{1,2,\ldots ,n\}\setminus \{i_1,i_2,\ldots ,i_{k-2}\}$ and which is planar with respect to the ordering $(1,2,\cdots, \slashed{i}_1,\cdots ,\slashed{i}_2,\cdots ,\slashed{i}_{k-2},\cdots, n)$ satisfying the following conditions
\begin{itemize}
    \item  Diagonal entries are the empty tree ${\cal A}_{\ldots ,i,\ldots, i,\ldots }=\emptyset$.
    \item  Compatibility: $d_{i_1i_2}^{(i_3,\ldots ,i_{k})}$ is completely symmetric in all $k$ indices. 
\end{itemize}
\end{defn}

A point which has not been explained so far is why each element in a collection, matrix or in general an array is called a Feynman diagram. We now turn to this point. The contribution to an amplitude of a given planar array of Feynman diagrams is computed using the function
\be
{\cal F}({\cal A}) = \sum_{i_1,i_2,\ldots ,i_{k}} \sfs_{i_1i_2\cdots i_k} d_{i_1i_2}^{(i_3,\ldots ,i_{k})}.
\ee
For $k=2$ it is easy to show that this function is independent of the external edge's lengths by writing $d_{ij}=e_i+e_j+d_{ij}^{\rm internal}$ and using momentum conservation. For $k=3$ it was noted in \cite{Borges:2019csl} that the function ${\cal F}({\cal C})$ can also be written in a way that it is also independent of the external edges. However, the proof is not as straightforward. In order to easily see this property all we have to do is to treat each tree in the array as a true Feynman diagram with its own kinematics. 

The element in the array ${\cal A}_{i_1,i_2\ldots ,i_{k-2}}$ is an $(n-k+2)$-particle Feynman diagram with particle labels $\{ 1,2,\ldots ,n\}\setminus \{i_1,i_2\ldots ,i_{k-2}\}$. As such, one has to associate the proper kinematic invariants satisfying momentum conservation. Let us introduce the notation ${\cal I}:= \{i_1,i_2\ldots ,i_{k-2}\}$ and $\overline{\cal I}$ for its complement. Then we have
\be
s^{({\cal I})}_{ii} = 0 , \quad \sum_{j\in \overline{\cal I}} s^{({\cal I})}_{ij} = 0 \quad \forall i\in \overline{\cal I}.
\ee
Using these kinematic invariants one can parametrize the $(k,n)$ invariants as 
\be
\sfs_{i_1i_2\ldots i_k}:= \sum_{{\cal I}\cup \{j_1,j_2\}=\{i_1,i_2,\ldots ,i_k\}} s_{j_1,j_2}^{({\cal I})},
\ee
where the sum is over all possible ways of decomposing $\{i_1,i_2,\ldots ,i_k\}$ into two sets of $k-2$ and $2$ elements respectively. To illustrate the notation consider $k=3$ where
\be\label{clu}
\sfs_{ijk} := s^{(i)}_{jk} + s^{(j)}_{ki} + s^{(k)}_{ij}.
\ee
This parametrization is very redundant but as any good redundancy it makes at least one property of the relevant object manifest. In this case it is the independence of the external edges of ${\cal F}({\cal A})$. Let us continue with the $k=3$ case in order not to clutter the notations but the general $k$ version is clear.

Using \eqref{clu} one can write
\be
{\cal F}({\cal C}) = \sum_{i,j,k}\sfs_{ijk} d^{(i)}_{jk}
\ee
as 
\be\label{polu}
{\cal F}({\cal C}) = \frac{1}{3}\sum_{i=1}^n \sum_{j,k} s^{(i)}_{jk} d^{(i)}_{jk}.
\ee
Here we used the symmetry property of $d^{(i)}_{jk}$ to identify all three terms coming from using \eqref{clu}. The new form is nothing but a sum over the functions $F(T)$ for each of the trees in the collection and therefore it is clearly independent of the external edges as expected. 

Let us now discuss how the two kinds of combinatorial bootstraps work for general planar arrays of Feynman diagrams.

The first kind of combinatorial bootstrap, which we called pruning-mutating in section \ref{sec2}, is simply the process of producing $C_{n-2}$ initial arrays of Feynman diagrams by starting with any given $n$-point planar Feynman diagram and pruning $k-2$ of its leaves in all possible ways to end up with an array of $(n-k+2)$-point Feynman diagrams. Starting from these initial planar arrays, one computes the corresponding metrics and find all their possible degenerations. Approaching each degeneration one at a time one can produce a new planar array by resolving the degeneration only in the other planar possible way. Repeating the mutation procedure on all new arrays generated until no new array is found leads to the full set of planar arrays of Feynman diagrams.

The second kind of combinatorial bootstrap, as described in section \ref{sec3} for planar matrices, is the idea that the compatibility conditions on the metrics of the trees making the array force it to be completely symmetric. This simple observation together with the fact that any subarray where some indices are fixed must in itself be a valid planar array of Feynman diagrams for some smaller values of $k$ and $n$ gives strong constrains on the objects.

As it should be clear from the examples presented in section \ref{sec3}, the second bootstrap approach is more efficient than the first one if all planar arrays in $(k-1,n-1)$ are known. This means that one could start with $(3,6)$ and produce the following sequence:
\be
(3,6)\to (4,7) \to (5,8) \to (6,9) \to (7,10) \ldots 
\ee
The reason to consider this sequence is that after obtaining all its elements, one can construct all $(3,n)$ planar collections via duality. Of course, in order to do that efficiently one has to find a combinatorial way of performing the duality directly at the level of the graphs.

\subsection{Combinatorial Duality} \label{sec61}

Let us start by defining some notation that will be used in this section. We will denote $T_n$ as a planar tree in $(2,n)$, ${\cal C}_n$ as a planar collection in $(3,n)$ and ${\cal M}_n$ as a planar matrix in $(4,n)$. In general, a planar array ${\cal A}_n$ will correspond to a $(k-2)$-dimensional array with dimensions of size $n$. 
In order to understand how the combinatorial duality works, we also introduce the concept of combinatorial soft limit. The combinatorial soft limit for particle $i$ applied to ${\cal A}_n$ is defined by removing the $i$-th $(k-3)$-dimensional array from ${\cal A}_n$, as well as removing the $i$-th label to the remaining $(k-3)$-dimensional arrays. Therefore, the combinatorial soft limit takes us from $(k,n)\to(k,n-1)$. 

It is useful to introduce a superscript ${\cal A}_n^{(i)}$ to refer to an array obtained from a combinatorial soft limit for particle $i$. Notice that this notation slightly differs from the one we use in earlier sections.

With this in hand, we can define the particular duality $(2,n)\sim(n-2,n)$ as taking the tree $T_n$ of $(2,n)$ and applying the combinatorial soft limit to particles $i_1,...,i_{n-2}$ in order to remove $n-2$ leaves to obtain the corresponding dual ${\cal A}^{(i_1,...,i_{n-2})}_{n-2}$ of $(n-2,n)$.

For general $(k,n)$ with $k<n-2$ the duality works as follows. Consider a $(k-2)$-dimensional planar array ${\cal A}_n$ of $(k,n)$. By taking the combinatorial soft limit for particle $i$, we end up with ${\cal A}_{n-1}^{(i)}$ of $(k,n-1)$. Apply this step $n$ times for all the $n$ particles. Now dualize each of the $n$ objects to directly obtain the corresponding $(n-k-2)$-dimensional array ${\cal A}_n$ of $(n-k,n)$, hence the duality. The combinatorial duality can be simply summarized as

\be\label{wer}
(k,n)\xrightarrow[\text{limit}]{\text{soft}} n\times (k,n-1)\xrightarrow[\text{dualize}]{\text{}}(n-k,n)
\ee

\subsubsection{Illustrative example: $(3,7)\sim(4,7)$} \label{sec611}

Now we proceed to show the explicit example for $(3,7)\sim(4,7)$. The combinatorial soft limit for particle $i$ applied to a planar collection ${\cal C}_n$\ corresponds to removing the $i$-th tree in ${\cal C}_n$ as well as removing the $i$-th label in all the rest of the trees in ${\cal C}_n$. Therefore, it implies ${\cal C}_{n}\to{\cal C}_{n-1}^{(i)}$. Similarly, the combinatorial soft limit for particle $i$ applied to a planar matrix ${\cal M}_n$ corresponds to removing the $i$-th column and row in ${\cal M}_n$ as well as removing the $i$-th label in all the rest of the trees in ${\cal M}_n$. Therefore, it implies ${\cal M}_{n}\to{\cal M}_{n-1}^{(i)}$.

Before studying $(3,7)\sim(4,7)$ let us consider $(3,6)\sim(3,6)$ as this will be useful below. Using \eqref{wer} we can see % 
\be
(3,6)\xrightarrow[\text{limit}]{\text{soft}} 6\times (3,5)\xrightarrow[\text{dualize}]{\text{}}(3,6)
\ee
where the duality $(2,5)\sim(3,5)$ is one of the most basic ones which was used as a motivation for introducing planar collections in \cite{Borges:2019csl}.

Now consider one planar collection ${\cal C}_{n=7}$ of $(3,7)$. By taking the combinatorial soft limit for particle $i$, we end up with a collection ${\cal C}_{n=6}^{(i)}$ in $(3,6)$. Given that $(3,6)\sim(3,6)$, this collection is dual to another collection $\tilde{{\cal C}}_{n=6}^{(i)}$, which corresponds to the $i$-th column of a planar matrix ${\cal M}_{n=7}$ in $(4,7)$. This means that if we now take the combinatorial soft limit for the other particles in ${\cal C}_{n=7}$ we end up with the full matrix ${\cal M}_{n=7}$. Hence, the objects ${\cal C}_{n=7}$ and ${\cal M}_{n=7}$ are dual.

We can also see this by following an equivalent path. Consider one planar matrix ${\cal M}_{n=7}$ of $(4,7)$. By taking the combinatorial soft limit for particle $i$, we end up with a planar matrix ${\cal M}_{n=6}^{(i)}$ of $(4,6)$. Notice that this matrix is dual to the planar tree $T_{n=6}^{(i)}$ of $(2,6)$ which is an element of ${\cal C}_{n=7}$, so by repeating the soft limit for all the remaining particles we end up with the full ${\cal C}_{n=7}$ of $(3,7)$.

%\begin{figure}[!htb]
%\centering
%\includegraphics[width=100mm]{3747.png}
%\caption{Combinatorial duality between $(3,7)$ and $(4,7)$. Each element of the graph -whether a planar tree, a planar collection or a planar matrix- goes together with the $(k,n)$ it corresponds to.}
%\label{3747}
%\end{figure}

\section{Future Directions} \label{sec7}

Generalized biadjoint amplitudes as defined by a CHY integral over the configuration space of $n$ points in $\mathbb{CP}^{k-2}$ with $k>2$ provide a very natural step beyond standard quantum field theory \cite{Cachazo:2019ngv}. An equally natural generalization of quantum field theory amplitudes is obtained by first identifying standard Feynman diagrams with metric trees and their connection to ${\rm Trop}\, G(2,n)$. In \cite{herrmann2009draw}, arrangements of metric trees where introduced as objects corresponding to ${\rm Trop}\, G(3,n)$. A special class of such arrangement, called planar collections of Feynman diagrams were then proposed as the simplest generalization of Feynman diagrams in \cite{Borges:2019csl}. In this work we introduced $(k-2)$-dimensional planar arrays of Feynman diagrams as the all $k$ generalization. One of the most exciting phenomena is that these $(k-2)$-dimensional arrays define generalized biadjoint amplitudes. 

The fact that both definitions of generalized amplitudes, either as a CHY integral or as a sum over arrays, coincide is non-trivial. In fact, a rigorous proof of this connection, perhaps along the lines of the proof for $k=2$ given by Dolan and Goddard \cite{Dolan:2013isa,Dolan:2014ega}, is a pressing problem. One possible direction is hinted by the observations made in section \ref{sec6}, where each Feynman diagram in an array was given its own kinematics along with its own metric. Of course, what makes the planar array interesting is the compatibility conditions for the metrics of the various trees in the array. Understanding the physical meaning of such conditions is also a very important problem. However, this already gives a hint as to what to do with the CHY integral. Borrowing the $k=3$ example in section \ref{sec6}, the kinematics is parameterized as $\sfs_{ijk} =s^{(i)}_{jk}+s^{(j)}_{ik}+s^{(k)}_{ij}$. Recall that in the CHY formulation on $\mathbb{CP}^{2}$ introduced in \cite{Cachazo:2019ngv} one starts with a potential function
\be
{\cal S}^{(3)}_n := \sum_{i,j,k} \sfs_{ijk} {\rm log}\, |ijk|
\ee
with $|ijk|$ Pl\"ucker coordinates in $G(3,n)$. Even though the object is antisymmetric in all its indices, only its absolute value is relevant in ${\cal S}^{(3)}_n$ since the way it enters in the CHY formula is only via the equations needed for the computation of its critical points. This means that $|ijk|$ can be used to define ``effective'' $k=2$  Pl\"ucker coordinates of the form $|jk|^{(i)}:= |ijk|$. In other words, once a label is selected, say $i$, then all other points in $\mathbb{CP}^{2}$ can be projected onto a $\mathbb{CP}^{1}$ using the $i^{\rm th}$-point. This means that the potential  ${\cal S}^{(3)}_n$  can be written as a sum over $n$ $k=2$ potentials in a way completely analogous to ${\cal F}(\cal C)$ in \eqref{polu}, i.e.
\be
{\cal S}^{(3)}_n = \frac{1}{3}\sum_{i=1}^n\, \sum_{j,k} s_{jk}^{(i)}\, |jk|^{(i)}.
\ee
One can then write a $k=3$ CHY formula as a product over $n$ $k=2$ CHY integrals linked by the ``compatibility constraints'' imposing that the absolute value of $|jk|^{(i)}$, $|ij|^{(k)}$, $|ik|^{(j)}$ all be equal. Owing to the techniques developed in \cite{Cachazo:2019ble,Agostini:2021rze,Sturmfels:2020mpv,Cachazo:2020uup,Cachazo:2020wgu}, many non-trivial 
$k=3,4$ CHY formulas have been verified to match the partial amplitudes obtained by using planar collections or matrices of Feynman diagrams, but general analysis just as what we present is still needed to understand the most general cases.

As explored in 
\cite{Guevara:2020lek}, by forcing a given planar array of Feynman diagrams to explore its degenerations of highest codimension one finds planar arrays of degenerate Feynman diagrams which encode the information of the poles  of the contributions of this planar array to the amplitudes. 
% It is worth studying whether all   planar arrays already contain all poles of generalized Feynman diagrams. 
How to connect CEGM amplitudes to 
cluster algebras  
\cite{SpeyerW,Drummond:2019qjk,Arkani-Hamed:2020tuz,Arkani-Hamed:2019plo,He:2021zuv,Drummond:2020kqg,Gates:2021tnp,Henke:2021ity},
positroid subdivisions  \cite{Lukowski:2020dpn,Early:2019eun,Early:2019zyi}, stringy integrals \cite{Arkani-Hamed:2019mrd,He:2020ray} or even the symbol alphabet of $\mathcal{N}=4$ SYM \cite{Henke:2019hve,Arkani-Hamed:2019rds} especially via their poles also deserves further exploring.

\vskip0.1in

{\bf Note Added:}

While the first version of this manuscript was being prepared for submission, the works \cite{Drummond:2019cxm,Arkani-Hamed:2019rds,Henke:2019hve,Arkani-Hamed:2019mrd} appeared which have some overlap with our results, especially in $(4,8)$.

While the third version of this paper was being prepared, some new results on local planarity appeared \cite{Cachazo:2022pnx,Cachazo:2023ltw}.
In this paper, we have studied an array of Feynman diagrams  consistent with a global notion of  planarity, 
 which is closely related to the positive part of the tropical Grassmannian, ${\rm Trop}\, G^+(k,n)$.
The notion of generalized Feynman diagrams was first introduced in \cite{Cachazo:2019ngv} and refined in \cite{Cachazo:2022pnx,Cachazo:2023ltw} to formalize the notion of local planarity or 
generalized color ordering.
  % which is a generalization of a collection or matrix of Feynman diagrams in this paper. 
%   A generalized Feynman diagram
% only satisfies local planarity as they are only required to be compatible with a generalized color ordering whose entries are  lower point standard $k=2$ orderings but are not necessary to be descendants of a universal $k=2$ ordering.
Generalized Feynman diagrams are expected to relate to the whole the tropical Grassmannian, ${\rm Trop}\, G(k,n)$, and it would be interesting to see how many properties of planar arrays of Feynman diagrams still hold there.

% Finally, another fascinating direction which we did not explore in this work is the fact that by forcing a given planar array of Feynman diagrams to explore its degenerations of highest codimension one finds planar arrays of degenerate Feynman diagrams which encode the information of the poles of the generalized amplitudes. The study of the structure of poles is of fundamental importance for uncovering the possible physical meaning of these amplitudes. We leave this for future work \cite{Guevara:2020lek}.

\section*{Acknowledgements}
%\acknowledgments

We would like to thank Nick Early and Song He for useful discussions. Research at Perimeter Institute is supported in part by the Government of Canada through the Department of Innovation, Science and Economic Development Canada and by the Province of Ontario through the Ministry of Economic Development, Job Creation and Trade.

\appendix

\section{Proof of One-to-one Map of a Binary Tree and its Metric} \label{apA}

\begin{lem}
\label{tatb}
Given that two cubic trees $T_A$ and $T_B$ have the same valid non-degenerate metric $d_{ij}$, then $T_A=T_B$.
\end{lem}

\begin{proof}
We are going to provide a proof by induction. First, consider the base case where $T_A$ and $T_B$ are $3$-point trees. It is clear that there exists a unique solution to $d_{12}=e_1+e_2$, $d_{13}=e_1+e_3$ and $d_{23}=e_2+e_3$. Since $T_A$ and $T_B$ have the same non-degenerate metric, the lengths $e_i^{(A)}=e_i^{(B)}$ must be identical, thus $T_A=T_B$. 

Now let us assume that the lemma is true for all $(n-1)$-point cubic metric trees and consider two $n$-point cubic trees $T_A$ and $T_B$ that have the same non-degenerate metric $d_{ij}$.
Next let us find leaves $i$ and $j$ such that $d_{il}-d_{jl}$ is $l$ independent. Such pair must exist because the condition is true for any pair of leaves which belong to the same ``cherry'' as shown in the diagrams in figure \ref{proof1}. Moreover, only leaves in cherries satisfy this condition in a cubic non-degenerate tree.

\begin{figure}[!htb]
\centering
\includegraphics[width=130mm]{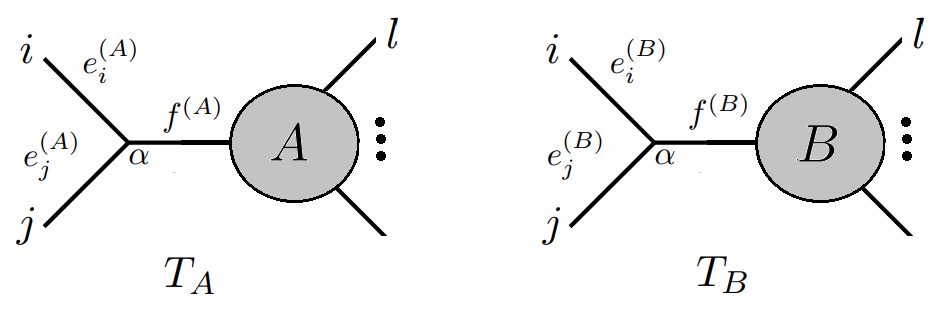}
\caption{Two $n$-point cubic trees with pairs $i$ and $j$ joined by the vertex $\alpha$.}
\label{proof1}
\end{figure}

Removing the cherries from both trees and introducing a new leaf $\alpha$ one can define a metric for the the $(n-1)$-point cubic trees in figure \ref{proof2}, whose leaves are given by $\left(\{1,2,\ldots ,n\}\setminus \{i,j\}\right) \cup\{\alpha\}$. 

\begin{figure}[!htb]
\centering
\includegraphics[width=130mm]{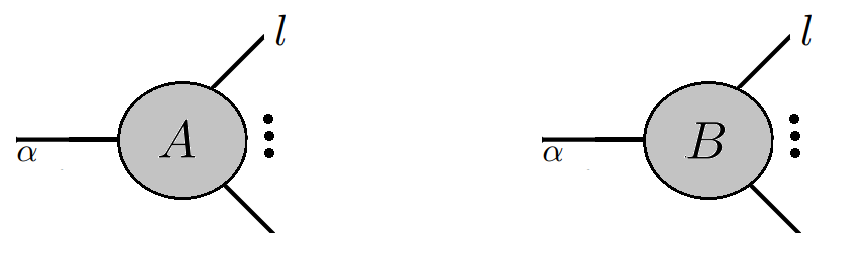}
\caption{Two $(n-1)$-point cubic trees with external edges $e_{\alpha}^{(A)}=f^{(A)}$ and $e_{\alpha}^{(B)}=f^{(B)}$ such that $d_{kl}^{(A)}=d_{kl}^{(B)}$.}
\label{proof2}
\end{figure}

Such a metric is defined in terms of the metric of the parent trees as follows. $d^{(A)}_{kl}=d_{kl}$ if $k,l\neq \alpha$ and $d^{(A)}_{k\alpha}=d_{ki}-e^{(A)}_i$. Likewise $d^{(B)}_{kl}=d_{kl}$ if $k,l\neq \alpha$ and $d^{(B)}_{k\alpha}=d_{ki}-e^{(B)}_i$. It is easy to see from the figure that the two metrics are identical, i.e. $d_{kl}^{(A)}=d_{kl}^{(B)}$. 

Using the induction hypothesis, the two metric trees in figure \ref{proof2} must be the same. In order to complete the proof all we need is to show that  $e_i^{(A)}=e_i^{(B)}$ and  $e_j^{(A)}=e_j^{(B)}$. The fact that $d_{il}=e_i^{(A)}+d_{\alpha l}^{(A)}=e_i^{(B)}+d_{\alpha l}^{(B)}$ immediately implies $e_i^{(A)}=e_i^{(B)}$, hence $T_A=T_B$.
\end{proof}

\section{All Planar Collections of Feynman Diagrams for $(3,6)$ } \label{apB}

Below we reproduce for the reader's convenience Table 1 of \cite{Borges:2019csl} which contains all $48$ planar collections of Feynman diagrams for $(3,6)$.  
The notation in this case is very compact and requires some explanation. Each collection for $(3,6)$ is made out of $5$-point trees. The tree in the $i^{\rm th}$-position must be planar with respect to the ordering $(1,2,\ldots ,\slashed{i},\ldots ,n)$. There is a single topology of five-point trees, i.e. a caterpillar tree with two cherries and one leg. Therefore it is possible to specify it by giving the label of the leaf attached to the leg. Using this, each collection becomes a one-dimensional array of six numbers.  

\begin{table}[H]
\centering
\begin{tabular}{|c|c||c|c|}
\hline 
\multicolumn{4}{|c|}{Planar collections of trees in $k=3$ and $n=6$} \\ 
\hline 
Collection & Trees & Collection & Trees \\ 
\hline 
$\mathcal{C}_{1}$ & $\qty[4,4,4,3,3,3]$ & $\mathcal{C}_{25}$ & $\qty[6,6,6,5,4,1]$ \\ 
\hline 
$\mathcal{C}_{2}$ & $\qty[4,4,4,3,6,5]$ & $\mathcal{C}_{26}$ & $\qty[6,6,6,6,6,3]$ \\ 
\hline 
$\mathcal{C}_{3}$ & $\qty[4,4,4,3,2,2]$ & $\mathcal{C}_{27}$ & $\qty[6,6,6,1,1,1]$ \\ 
\hline 
$\mathcal{C}_{4}$ & $\qty[4,4,4,1,4,4]$ & $\mathcal{C}_{28}$ & $\qty[6,6,6,2,2,1]$ \\ 
\hline 
$\mathcal{C}_{5}$ & $\qty[4,4,4,1,1,1]$ & $\mathcal{C}_{29}$ & $\qty[6,3,2,5,4,1]$ \\ 
\hline 
$\mathcal{C}_{6}$ & $\qty[4,4,6,6,6,5]$ & $\mathcal{C}_{30}$ & $\qty[6,3,2,1,1,1]$ \\ 
\hline 
$\mathcal{C}_{7}$ & $\qty[4,4,6,6,2,2]$ & $\mathcal{C}_{31}$ & $\qty[6,3,2,2,2,1]$ \\ 
\hline 
$\mathcal{C}_{8}$ & $\qty[4,5,5,5,4,4]$ & $\mathcal{C}_{32}$ & $\qty[2,5,5,5,2,2]$ \\ 
\hline 
$\mathcal{C}_{9}$ & $\qty[4,6,6,5,4,4]$ & $\mathcal{C}_{33}$ & $\qty[2,5,2,2,2,2]$ \\ 
\hline 
$\mathcal{C}_{10}$ & $\qty[4,6,6,2,2,4]$ & $\mathcal{C}_{34}$ & $\qty[2,1,4,3,3,3]$ \\ 
\hline 
$\mathcal{C}_{11}$ & $\qty[4,1,1,1,4,4]$ & $\mathcal{C}_{35}$ & $\qty[2,1,4,3,6,5]$ \\ 
\hline 
$\mathcal{C}_{12}$ & $\qty[4,1,1,1,1,1]$ & $\mathcal{C}_{36}$ & $\qty[2,1,4,3,2,2]$ \\ 
\hline 
$\mathcal{C}_{13}$ & $\qty[4,3,2,5,4,4]$ & $\mathcal{C}_{37}$ & $\qty[2,1,6,6,6,5]$ \\ 
\hline 
$\mathcal{C}_{14}$ & $\qty[4,3,2,2,2,4]$ & $\mathcal{C}_{38}$ & $\qty[2,1,6,6,2,2]$ \\ 
\hline 
$\mathcal{C}_{15}$ & $\qty[5,5,4,3,3,3]$ & $\mathcal{C}_{39}$ & $\qty[2,1,1,1,3,3]$ \\ 
\hline 
$\mathcal{C}_{16}$ & $\qty[5,5,4,3,6,5]$ & $\mathcal{C}_{40}$ & $\qty[2,1,1,1,6,5]$ \\ 
\hline 
$\mathcal{C}_{17}$ & $\qty[5,5,5,5,2,5]$ & $\mathcal{C}_{41}$ & $\qty[2,1,1,1,2,2]$ \\ 
\hline 
$\mathcal{C}_{18}$ & $\qty[5,5,6,6,6,5]$ & $\mathcal{C}_{42}$ & $\qty[3,5,5,5,4,3]$ \\ 
\hline 
$\mathcal{C}_{19}$ & $\qty[5,5,1,1,3,3]$ & $\mathcal{C}_{43}$ & $\qty[3,5,5,1,1,3]$ \\ 
\hline 
$\mathcal{C}_{20}$ & $\qty[5,5,1,1,6,5]$ & $\mathcal{C}_{44}$ & $\qty[3,3,6,3,3,3]$ \\ 
\hline 
$\mathcal{C}_{21}$ & $\qty[5,5,2,2,2,5]$ & $\mathcal{C}_{45}$ & $\qty[3,3,6,6,6,3]$ \\ 
\hline 
$\mathcal{C}_{22}$ & $\qty[6,5,5,5,4,1]$ & $\mathcal{C}_{46}$ & $\qty[3,3,2,5,4,3]$ \\ 
\hline 
$\mathcal{C}_{23}$ & $\qty[6,5,5,1,1,1]$ & $\mathcal{C}_{47}$ & $\qty[3,3,2,1,1,3]$ \\ 
\hline 
$\mathcal{C}_{24}$ & $\qty[6,6,6,3,3,3]$ & $\mathcal{C}_{48}$ & $\qty[3,3,2,2,2,3]$ \\ 
\hline 
\end{tabular}
\caption{All $48$ planar collections of trees for $n=6$ in a compact notation tailored to this case and explained in the text.}
\label{tb:n6coll}
\end{table}

\section{\label{polymake}Triangulation Functions in \textsc{PolyMake}}
In this section we show how to use the {\tt TRIANGULATION}
function of \textsc{PolyMake} needed in section \ref{sec5}.

In practice, the computation of
the vertices $\{V_{1},V_{2},V_{3},V_{4},V_{5}\}$ as well as the triangulation
that assigns them to two simplices $\Delta_{1}$ and $\Delta_{2}$
is automated by the software \textsc{PolyMake}. The only input needed is the
set of faces (\ref{eq:6z}). We now provide the corresponding
script to automate the process, given by the following command lines:

\begin{verbatim}
    open(INPUT,"<","boundaries.txt");
    
    $mtrr= new Array<String>(<INPUT>);close(INPUT);$t1=time;
    
    for(my $i=0;$i<scalar(@{$mtrr});$i++)
        {print "$i\n";@s1=split("X",$mtrr->[$i]);@arm=();
       
        for(my $j=0;$j<scalar(@s1);$j++)
            {@dst=split(",",$s1[$j]);
            $arm[$j]=new Vector<Rational>(@dst)};
       
        $planes=new Matrix<Rational>(@arm);
        $pol= new Polytope(INEQUALITIES=>$planes);
        
        open(my $f,">>","facets.txt");
        print $f $pol->TRIANGULATION->FACETS, "\n"; close $f;
        open(my $g,">>","vertices.txt");
        print $g $pol->VERTICES , "\n"; close $g;};
    
    $t2=time;print $t2-$t1;
\end{verbatim}

In order to implement the \textsc{PolyMake} script we need to create a file {\tt boundaries.txt}, in which each row corresponds to the planes $Z_i$'s of a collection. In a row, different vectors are separated by the character {\tt X} and vector components are separated by commas. For instance, for the bipyramid case \eqref{eq:6z}, the row reads: 
\begin{verbatim}
1,0,0,0X0,1,0,0X0,0,1,0X0,0,0,1X0,1,-1,1X1,0,-1,1
\end{verbatim}
An arbitrary number of collections (for instance to obtain the full amplitude) can be processed simply by adding rows into the {\tt .txt} file. The script will display a counter to indicate the row being processed, as well as the total time of the computation (in seconds) once it is completed.

The output of the script are two text files {\tt facets.txt} and {\tt vertices.txt}. For the previous example {\tt vertices.txt} contains
\begin{verbatim}
0 0 1 1
1 0 0 0
1 1 1 0
0 1 0 0
0 0 0 1
\end{verbatim}
This is just a list of vertices $V_i$ (note the unusual indexation starting from $i=0$):
\begin{eqnarray}
V_0=(0,0,1,1)\,,V_1=(1,0,0,0)\,,V_2=(1,1,1,0)\ldots
\end{eqnarray}
while {\tt facets.txt} contains
\begin{verbatim}
{0 1 2 3}
{0 1 3 4}
\end{verbatim}
meaning that the full object can be triangulated by two simplices, with vertices $\{V_0,V_1,V_2,V_3\}$ and $\{V_0,V_1,V_3,V_4\}$ respectively (these are relabellings of the previous simplices $\Delta_1$ and $\Delta_2$). The corresponding contribution to the amplitude is then given by the volume formula \eqref{eq:vt}.

\bibliographystyle{JHEP}
\bibliography{references}

\providecommand{\href}[2]{#2}\begingroup\raggedright\begin{thebibliography}{10}

\bibitem{Fairlie:1972zz}
D.~Fairlie and D.~Roberts, {\it Dual models without tachyons—a new approach},
   {\em unpublished Durham preprint PRINT-72-2440} {\bf 1972} (1972).

\bibitem{Fairlie:2008dg}
D.~B. Fairlie, {\it {A Coding of Real Null Four-Momenta into World-Sheet
  Coordinates}},  {\em Adv. Math. Phys.} {\bf 2009} (2009) 284689,
  [\href{http://arxiv.org/abs/0805.2263}{{\tt arXiv:0805.2263}}].

\bibitem{Cachazo:2013gna}
F.~Cachazo, S.~He, and E.~Y. Yuan, {\it {Scattering equations and
  Kawai-Lewellen-Tye orthogonality}},  {\em Phys. Rev.} {\bf D90} (2014), no.~6
  065001, [\href{http://arxiv.org/abs/1306.6575}{{\tt arXiv:1306.6575}}].

\bibitem{Cachazo:2013hca}
F.~Cachazo, S.~He, and E.~Y. Yuan, {\it {Scattering of Massless Particles in
  Arbitrary Dimensions}},  {\em Phys. Rev. Lett.} {\bf 113} (2014), no.~17
  171601, [\href{http://arxiv.org/abs/1307.2199}{{\tt arXiv:1307.2199}}].

\bibitem{Dolan:2013isa}
L.~Dolan and P.~Goddard, {\it {Proof of the Formula of Cachazo, He and Yuan for
  Yang-Mills Tree Amplitudes in Arbitrary Dimension}},  {\em JHEP} {\bf 05}
  (2014) 010, [\href{http://arxiv.org/abs/1311.5200}{{\tt arXiv:1311.5200}}].

\bibitem{Cachazo:2019ngv}
F.~Cachazo, N.~Early, A.~Guevara, and S.~Mizera, {\it {Scattering Equations:
  From Projective Spaces to Tropical Grassmannians}},  {\em JHEP} {\bf 06}
  (2019) 039, [\href{http://arxiv.org/abs/1903.08904}{{\tt arXiv:1903.08904}}].

\bibitem{speyer2004tropical}
D.~Speyer and B.~Sturmfels, {\it The tropical grassmannian},  {\em Advances in
  Geometry} {\bf 4} (2004), no.~3 389--411.

\bibitem{speyer2005tropical}
D.~Speyer and L.~Williams, {\it The tropical totally positive grassmannian},
  {\em Journal of Algebraic Combinatorics} {\bf 22} (2005), no.~2 189--210.

\bibitem{Cachazo:2019apa}
F.~Cachazo and J.~M. Rojas, {\it {Notes on Biadjoint Amplitudes, ${\rm
  Trop}\,G(3,7)$ and $X(3,7)$ Scattering Equations}},  {\em JHEP} {\bf 04}
  (2020) 176, [\href{http://arxiv.org/abs/1906.05979}{{\tt arXiv:1906.05979}}].

\bibitem{Sepulveda:2019vrz}
D.~García~Sepúlveda and A.~Guevara, {\it {A Soft Theorem for the Tropical
  Grassmannian}},  \href{http://arxiv.org/abs/1909.05291}{{\tt
  arXiv:1909.05291}}.

\bibitem{Cachazo:2019ble}
F.~Cachazo, B.~Umbert, and Y.~Zhang, {\it {Singular Solutions in Soft Limits}},
   {\em JHEP} {\bf 05} (2020) 148, [\href{http://arxiv.org/abs/1911.02594}{{\tt
  arXiv:1911.02594}}].

\bibitem{GarciaSepulveda:2019jxn}
D.~Garc\'\i{}a~Sep\'ulveda and A.~Guevara, {\it {A Soft Theorem for the
  Tropical Grassmannian}},  \href{http://arxiv.org/abs/1909.05291}{{\tt
  arXiv:1909.05291}}.

\bibitem{Abhishek:2020xfy}
M.~Abhishek, S.~Hegde, D.~P. Jatkar, and A.~P. Saha, {\it {Double soft theorem
  for generalised biadjoint scalar amplitudes}},  {\em SciPost Phys.} {\bf 10}
  (2021), no.~2 036, [\href{http://arxiv.org/abs/2008.07271}{{\tt
  arXiv:2008.07271}}].

\bibitem{Early:2022mdn}
N.~Early, {\it {Factorization for Generalized Biadjoint Scalar Amplitudes via
  Matroid Subdivisions}},  \href{http://arxiv.org/abs/2211.16623}{{\tt
  arXiv:2211.16623}}.

\bibitem{Cachazo:2021wsz}
F.~Cachazo, N.~Early, and B.~Gim\'enez~Umbert, {\it {Smoothly splitting
  amplitudes and semi-locality}},  {\em JHEP} {\bf 08} (2022) 252,
  [\href{http://arxiv.org/abs/2112.14191}{{\tt arXiv:2112.14191}}].

\bibitem{herrmann2009draw}
S.~Herrmann, A.~Jensen, M.~Joswig, and B.~Sturmfels, {\it How to draw tropical
  planes},  {\em the electronic journal of combinatorics} {\bf 16} (2009),
  no.~2 6.

\bibitem{Borges:2019csl}
F.~Borges and F.~Cachazo, {\it {Generalized Planar Feynman Diagrams:
  Collections}},  {\em JHEP} {\bf 11} (2020) 164,
  [\href{http://arxiv.org/abs/1910.10674}{{\tt arXiv:1910.10674}}].

\bibitem{SpeyerW}
D.~{Speyer} and L.~K. {Williams}, {\it {The tropical totally positive
  Grassmannian}},  {\em arXiv Mathematics e-prints} (Dec, 2003) math/0312297,
  [\href{http://arxiv.org/abs/math/0312297}{{\tt math/0312297}}].

\bibitem{Drummond:2019qjk}
J.~Drummond, J.~Foster, O.~G\"urdogan, and C.~Kalousios, {\it {Tropical
  Grassmannians, cluster algebras and scattering amplitudes}},  {\em JHEP} {\bf
  04} (2020) 146, [\href{http://arxiv.org/abs/1907.01053}{{\tt
  arXiv:1907.01053}}].

\bibitem{ClusterA}
S.~Fomin and A.~Zelevinsky, {\it Cluster algebras i: foundations},  {\em
  Journal of the American mathematical society} {\bf 15} (2002), no.~2
  497--529.

\bibitem{ClusterB}
S.~{Fomin} and A.~{Zelevinsky}, {\it {Cluster algebras II: Finite type
  classification}},  {\em Inventiones Mathematicae} {\bf 154} (Oct, 2003)
  63--121, [\href{http://arxiv.org/abs/math/0208229}{{\tt math/0208229}}].

\bibitem{ClusterC}
A.~{Berenstein}, S.~{Fomin}, and A.~{Zelevinsky}, {\it {Cluster algebras III:
  Upper bounds and double Bruhat cells}},  {\em arXiv Mathematics e-prints}
  (May, 2003) math/0305434, [\href{http://arxiv.org/abs/math/0305434}{{\tt
  math/0305434}}].

\bibitem{Guevara:2020lek}
A.~Guevara and Y.~Zhang, {\it {Planar Matrices and Arrays of Feynman Diagrams:
  Poles for Higher $k$}},  \href{http://arxiv.org/abs/2007.15679}{{\tt
  arXiv:2007.15679}}.

\bibitem{Frost:2018djd}
H.~Frost, {\it {Biadjoint scalar tree amplitudes and intersecting dual
  associahedra}},  {\em JHEP} {\bf 06} (2018) 153,
  [\href{http://arxiv.org/abs/1802.03384}{{\tt arXiv:1802.03384}}].

\bibitem{He:2020ray}
S.~He, L.~Ren, and Y.~Zhang, {\it {Notes on polytopes, amplitudes and boundary
  configurations for Grassmannian string integrals}},  {\em JHEP} {\bf 04}
  (2020) 140, [\href{http://arxiv.org/abs/2001.09603}{{\tt arXiv:2001.09603}}].

\bibitem{Arkani-Hamed:2017tmz}
N.~Arkani-Hamed, Y.~Bai, and T.~Lam, {\it {Positive Geometries and Canonical
  Forms}},  {\em JHEP} {\bf 11} (2017) 039,
  [\href{http://arxiv.org/abs/1703.04541}{{\tt arXiv:1703.04541}}].

\bibitem{Arkani-Hamed:2019mrd}
N.~Arkani-Hamed, S.~He, and T.~Lam, {\it {Stringy canonical forms}},  {\em
  JHEP} {\bf 02} (2021) 069, [\href{http://arxiv.org/abs/1912.08707}{{\tt
  arXiv:1912.08707}}].

\bibitem{Arkani-Hamed:2017mur}
N.~Arkani-Hamed, Y.~Bai, S.~He, and G.~Yan, {\it {Scattering Forms and the
  Positive Geometry of Kinematics, Color and the Worldsheet}},  {\em JHEP} {\bf
  05} (2018) 096, [\href{http://arxiv.org/abs/1711.09102}{{\tt
  arXiv:1711.09102}}].

\bibitem{Dolan:2014ega}
L.~Dolan and P.~Goddard, {\it {The Polynomial Form of the Scattering
  Equations}},  {\em JHEP} {\bf 07} (2014) 029,
  [\href{http://arxiv.org/abs/1402.7374}{{\tt arXiv:1402.7374}}].

\bibitem{Agostini:2021rze}
D.~Agostini, T.~Brysiewicz, C.~Fevola, L.~K\"uhne, B.~Sturmfels, S.~Telen, and
  T.~Lam, {\it {Likelihood degenerations}},  {\em Adv. Math.} {\bf 414} (2023)
  108863, [\href{http://arxiv.org/abs/2107.10518}{{\tt arXiv:2107.10518}}].

\bibitem{Sturmfels:2020mpv}
B.~Sturmfels and S.~Telen, {\it {Likelihood Equations and Scattering
  Amplitudes}},  \href{http://arxiv.org/abs/2012.05041}{{\tt
  arXiv:2012.05041}}.

\bibitem{Cachazo:2020uup}
F.~Cachazo and N.~Early, {\it {Minimal Kinematics: An All $k$ and $n$ Peek into
  ${\rm Trop}^+{\rm G}(k,n)$}},  {\em SIGMA} {\bf 17} (2021) 078,
  [\href{http://arxiv.org/abs/2003.07958}{{\tt arXiv:2003.07958}}].

\bibitem{Cachazo:2020wgu}
F.~Cachazo and N.~Early, {\it {Planar Kinematics: Cyclic Fixed Points, Mirror
  Superpotential, k-Dimensional Catalan Numbers, and Root Polytopes}},
  \href{http://arxiv.org/abs/2010.09708}{{\tt arXiv:2010.09708}}.

\bibitem{Arkani-Hamed:2020tuz}
N.~Arkani-Hamed, S.~He, and T.~Lam, {\it {Cluster Configuration Spaces of
  Finite Type}},  {\em SIGMA} {\bf 17} (2021) 092,
  [\href{http://arxiv.org/abs/2005.11419}{{\tt arXiv:2005.11419}}].

\bibitem{Arkani-Hamed:2019plo}
N.~Arkani-Hamed, S.~He, T.~Lam, and H.~Thomas, {\it {Binary geometries,
  generalized particles and strings, and cluster algebras}},  {\em Phys. Rev.
  D} {\bf 107} (2023), no.~6 066015,
  [\href{http://arxiv.org/abs/1912.11764}{{\tt arXiv:1912.11764}}].

\bibitem{He:2021zuv}
S.~He, Y.~Wang, Y.~Zhang, and P.~Zhao, {\it {Notes on worldsheet-like variables
  for cluster configuration spaces}},
  \href{http://arxiv.org/abs/2109.13900}{{\tt arXiv:2109.13900}}.

\bibitem{Drummond:2020kqg}
J.~Drummond, J.~Foster, O.~G\"urdo\u{g}an, and C.~Kalousios, {\it {Tropical
  fans, scattering equations and amplitudes}},  {\em JHEP} {\bf 11} (2021) 071,
  [\href{http://arxiv.org/abs/2002.04624}{{\tt arXiv:2002.04624}}].

\bibitem{Gates:2021tnp}
S.~J. Gates, S.~N. Hazel~Mak, M.~Spradlin, and A.~Volovich, {\it {Cluster
  Superalgebras and Stringy Integrals}},
  \href{http://arxiv.org/abs/2111.08186}{{\tt arXiv:2111.08186}}.

\bibitem{Henke:2021ity}
N.~Henke and G.~Papathanasiou, {\it {Singularities of eight- and nine-particle
  amplitudes from cluster algebras and tropical geometry}},  {\em JHEP} {\bf
  10} (2021) 007, [\href{http://arxiv.org/abs/2106.01392}{{\tt
  arXiv:2106.01392}}].

\bibitem{Lukowski:2020dpn}
T.~Lukowski, M.~Parisi, and L.~K. Williams, {\it {The positive tropical
  Grassmannian, the hypersimplex, and the m=2 amplituhedron}},
  \href{http://arxiv.org/abs/2002.06164}{{\tt arXiv:2002.06164}}.

\bibitem{Early:2019eun}
N.~Early, {\it {Planar kinematic invariants, matroid subdivisions and
  generalized Feynman diagrams}},  \href{http://arxiv.org/abs/1912.13513}{{\tt
  arXiv:1912.13513}}.

\bibitem{Early:2019zyi}
N.~Early, {\it {From weakly separated collections to matroid subdivisions}},
  \href{http://arxiv.org/abs/1910.11522}{{\tt arXiv:1910.11522}}.

\bibitem{Henke:2019hve}
N.~Henke and G.~Papathanasiou, {\it {How tropical are seven- and eight-particle
  amplitudes?}},  {\em JHEP} {\bf 08} (2020) 005,
  [\href{http://arxiv.org/abs/1912.08254}{{\tt arXiv:1912.08254}}].

\bibitem{Arkani-Hamed:2019rds}
N.~Arkani-Hamed, T.~Lam, and M.~Spradlin, {\it {Non-perturbative geometries for
  planar $ \mathcal{N} $ = 4 SYM amplitudes}},  {\em JHEP} {\bf 03} (2021) 065,
  [\href{http://arxiv.org/abs/1912.08222}{{\tt arXiv:1912.08222}}].

\bibitem{Drummond:2019cxm}
J.~Drummond, J.~Foster, O.~G\"urdogan, and C.~Kalousios, {\it {Algebraic
  singularities of scattering amplitudes from tropical geometry}},  {\em JHEP}
  {\bf 04} (2021) 002, [\href{http://arxiv.org/abs/1912.08217}{{\tt
  arXiv:1912.08217}}].

\bibitem{Cachazo:2022pnx}
F.~Cachazo, N.~Early, and Y.~Zhang, {\it {Color-Dressed Generalized Biadjoint
  Scalar Amplitudes: Local Planarity}},
  \href{http://arxiv.org/abs/2212.11243}{{\tt arXiv:2212.11243}}.

\bibitem{Cachazo:2023ltw}
F.~Cachazo, N.~Early, and Y.~Zhang, {\it {Generalized Color Orderings: CEGM
  Integrands and Decoupling Identities}},
  \href{http://arxiv.org/abs/2304.07351}{{\tt arXiv:2304.07351}}.

\end{thebibliography}\endgroup

\end{document}